\documentclass[copyright,creativecommons]{eptcs}
 
\usepackage{url}
\usepackage{hyperref}

\usepackage[all]{hypcap}

\usepackage{amsmath,amsthm,amssymb}
\usepackage{mathtools}
\usepackage{xspace,enumerate,epsfig}
\usepackage{graphicx}
\graphicspath{{.}{./figures/}}
\usepackage{marvosym}
\usepackage{bm}
\usepackage{tabu}
\usepackage[draft]{tikzit}

\input{circuits.tikzdefs}

\tikzstyle{gate}=[shape=rectangle, text height=1.5ex, text depth=0.25ex, yshift=0.5mm, fill=white, draw=black, minimum height=5mm, yshift=-0.5mm, minimum width=5mm, font={\small}, tikzit category=circuit]
\tikzstyle{big gate}=[shape=rectangle, text height=1.5ex, text depth=0.25ex, yshift=0.5mm, fill=white, draw=black, minimum height=10mm, yshift=-0.5mm, minimum width=5mm, font={\small}, tikzit category=circuit]
\tikzstyle{box}=[shape=rectangle, text height=1.5ex, text depth=0.25ex, yshift=0.5mm, fill=white, draw=black, minimum height=5mm, yshift=-0.5mm, minimum width=5mm, font={\footnotesize}]
\tikzstyle{medium box}=[shape=rectangle, text height=1.5ex, text depth=0.25ex, yshift=0.5mm, fill=white, draw=black, minimum height=10mm, yshift=-0.5mm, minimum width=5mm, font={\footnotesize}]
\tikzstyle{large box}=[shape=rectangle, text height=1.5ex, text depth=0.25ex, yshift=0.5mm, fill=white, draw=black, minimum height=14mm, yshift=-0.5mm, minimum width=5mm, font={\footnotesize}]
\tikzstyle{Z dot}=[inner sep=0mm, minimum size=2mm, shape=circle, draw=black, fill=zxgreen, tikzit fill={rgb,255: red,216; green,248; blue,216}]
\tikzstyle{Z phase dot}=[minimum size=5mm, font={\footnotesize}, shape=rectangle, rounded corners=2mm, inner sep=0.2mm, outer sep=-2mm, scale=0.8, tikzit shape=circle, draw=black, fill=zxgreen, tikzit draw=blue, tikzit fill={rgb,255: red,216; green,248; blue,216}]
\tikzstyle{Z bold dot}=[inner sep=0mm, minimum size=2mm, shape=circle, draw=black, fill=zxgreen, tikzit fill={rgb,255: red,216; green,248; blue,216}, line width=1.2pt, tikzit category=zx]
\tikzstyle{Z tiny dot}=[inner sep=0mm, minimum size=1mm, shape=circle, draw=black, fill=zxgreen, tikzit fill={rgb,255: red,216; green,248; blue,216}]
\tikzstyle{X dot}=[Z dot, shape=circle, draw=black, fill=zxred, tikzit fill={rgb,255: red,221; green,165; blue,165}]
\tikzstyle{X phase dot}=[Z phase dot, tikzit shape=circle, tikzit draw=blue, fill=zxred, font={\footnotesize}, tikzit fill={rgb,255: red,221; green,165; blue,165}]
\tikzstyle{X bold dot}=[inner sep=0mm, minimum size=2mm, shape=circle, draw=black, fill=zxred, tikzit fill={rgb,255: red,221; green,165; blue,165}, line width=1.2pt, tikzit category=zx]
\tikzstyle{X tiny dot}=[inner sep=0mm, minimum size=1mm, shape=circle, draw=black, fill=zxred, tikzit fill={rgb,255: red,221; green,165; blue,165}]
\tikzstyle{hadamard}=[fill=zxhad, draw=black, shape=rectangle, inner sep=0.6mm, minimum height=1.5mm, minimum width=1.5mm, tikzit fill=yellow, font={\footnotesize}]
\tikzstyle{paulibox}=[fill=zxpaulibox, draw=black, shape=rectangle, inner sep=0.6mm, minimum height=5mm, minimum width=5mm, font={\footnotesize}, text height=1.5ex, text depth=0.25ex, tikzit fill={rgb,255: red,138; green,159; blue,198}]
\tikzstyle{big paulibox}=[fill=zxpaulibox, draw=black, shape=rectangle, inner sep=0.6mm, minimum height=15mm, minimum width=5mm, font={\footnotesize}, text height=1.5ex, text depth=0.25ex, tikzit fill={rgb,255: red,138; green,159; blue,198}]
\tikzstyle{vertex}=[inner sep=0mm, minimum size=1mm, shape=circle, draw=black, fill=black]
\tikzstyle{vertex set}=[inner sep=0mm, minimum size=1mm, shape=circle, draw=black, fill=white, font={\footnotesize}]
\tikzstyle{small black dot}=[fill=black, draw=black, shape=circle, inner sep=0pt, minimum width=1.2mm, tikzit category=circuit]
\tikzstyle{cnot ctrl}=[fill=black, draw=black, shape=circle, inner sep=0pt, minimum width=1.2mm, tikzit category=circuit]
\tikzstyle{cnot targ}=[fill=white, draw=white, shape=circle, tikzit category=circuit, label={center:$\oplus$}, inner sep=0pt, minimum width=2.1mm, tikzit fill={rgb,255: red,102; green,204; blue,255}, tikzit draw=black]
\tikzstyle{left state}=[fill=white, draw=black, shape=regular polygon, regular polygon sides=3, regular polygon rotate=-30, scale=0.7, inner sep=0.5pt, minimum size=1cm]
\tikzstyle{ket}=[fill=white, draw=black, shape=regular polygon, regular polygon sides=3, regular polygon rotate=-30, scale=0.7, inner sep=0.5pt, minimum size=1cm]
\tikzstyle{right state}=[fill=white, draw=black, shape=regular polygon, regular polygon sides=3, regular polygon rotate=30, scale=0.7, inner sep=0.5pt, minimum size=1cm]
\tikzstyle{bra}=[fill=white, draw=black, shape=regular polygon, regular polygon sides=3, regular polygon rotate=30, scale=0.7, inner sep=0.5pt, minimum size=1cm]
\tikzstyle{wide left state}=[fill=white, draw=black, minimum height=8mm, minimum width=1.5cm, shape=isosceles triangle, isosceles triangle stretches, shape border rotate=180, scale=0.7, inner sep=1pt]
\tikzstyle{wide right state}=[fill=white, draw=black, minimum height=8mm, minimum width=1.5cm, shape=isosceles triangle, isosceles triangle stretches, scale=0.7, inner sep=1pt]
\tikzstyle{scalar}=[shape=rectangle, text height=1.5ex, text depth=0.25ex, yshift=0.5mm, fill=white, draw=black, minimum height=5mm, yshift=-0.5mm, minimum width=5mm, font={\small}]
\tikzstyle{pointer edge part}=[very thick, gray]
\tikzstyle{white dot}=[Z dot, tikzit fill={rgb,255: red,216; green,248; blue,216}]
\tikzstyle{white phase dot}=[Z phase dot, tikzit fill={rgb,255: red,216; green,248; blue,216}, tikzit draw=blue]
\tikzstyle{gray dot}=[X dot, tikzit fill={rgb,255: red,221; green,165; blue,165}]
\tikzstyle{gray phase dot}=[X phase dot, tikzit fill={rgb,255: red,221; green,165; blue,165}, tikzit draw=blue]
\tikzstyle{small hadamard}=[hadamard, tikzit shape=rectangle, tikzit fill=yellow]
\tikzstyle{clabel}=[fill=white, draw=none, shape=rectangle, tikzit fill={rgb,255: red,56; green,255; blue,242}, font={\footnotesize}, inner sep=1pt]
\tikzstyle{diag label}=[fill=white, draw=none, shape=circle, tikzit fill={rgb,255: red,56; green,255; blue,242}, inner sep=2pt]
\tikzstyle{empty diagram}=[draw={gray!40!white}, dashed, shape=rectangle, minimum width=1cm, minimum height=1cm]
\tikzstyle{empty diagram small}=[draw={gray!50!white}, dashed, shape=rectangle, minimum width=0.6cm, minimum height=0.5cm]
\tikzstyle{growbox}=[fill={rgb,255: red,250; green,196; blue,255}, draw=black, shape=Ebox, tikzit shape=rectangle, minimum height=6mm]
\tikzstyle{st}=[star, star points=5, fill=white, draw=black, inner sep=1.2pt, line width=1.2pt, tikzit fill=blue, tikzit draw=red, tikzit category=ZH-pf]

\tikzstyle{hadamard edge}=[-, dashed, dash pattern=on 2pt off 0.5pt, thick, draw={rgb,255: red,68; green,136; blue,255}]
\tikzstyle{hadamard edge alt}=[-, dashed, dash pattern=on 2pt off 0.5pt, thick, draw={rgb,255: red,150; green,200; blue,255}]
\tikzstyle{box edge}=[-, dashed, dash pattern=on 2pt off 0.5pt, thick, draw={rgb,255: red,203; green,192; blue,225}]
\tikzstyle{brace edge}=[-, tikzit draw=blue, decorate, decoration={brace,amplitude=1mm,raise=-1mm}]
\tikzstyle{diredge}=[->]
\tikzstyle{double edge}=[-, double, shorten <=-1mm, shorten >=-1mm, double distance=2pt]
\tikzstyle{gray edge}=[-, {gray!60!white}]
\tikzstyle{pointer edge}=[->, very thick, gray]
\tikzstyle{boldedge}=[-, line width=1.6pt, shorten <=-0.17mm, shorten >=-0.17mm]
\tikzstyle{simple}=[-]
\tikzstyle{dashed edge}=[-, dashed, dash pattern=on 2pt off 0.5pt, draw=black]
\tikzstyle{light dashed edge}=[-, dashed, dash pattern=on 2pt off 0.5pt, draw={rgb,255: red,191; green,191; blue,191}]
\tikzstyle{surface X}=[-, tikzit fill=red, fill=zxred]
\tikzstyle{surface Z}=[-, tikzit fill=green, fill=zxgreen]

\usepackage{stmaryrd}
\usepackage{keycommand}

\usepackage{ifthen}

\theoremstyle{definition}
\newtheorem{theorem}{Theorem}[section]

\newtheorem{lemma}[theorem]{Lemma}
\newtheorem{proposition}[theorem]{Proposition}

\newtheorem{definition}[theorem]{Definition}

\newtheorem{example}[theorem]{Example}

\newtheorem{example*}[theorem]{Example*}
\newtheorem{examples*}[theorem]{Examples*}
\newtheorem{remark}[theorem]{Remark}
\newtheorem{remark*}[theorem]{Remark*}

\usepackage[color]{changebar}

\usepackage{color}
\def\bR{\begin{color}{red}} 
\def\bB{\begin{color}{blue}}
\def\bM{\begin{color}{magenta}}
\def\bC{\begin{color}{cyan}}
\def\bW{\begin{color}{white}}
\def\bBl{\begin{color}{black}} 
\def\bG{\begin{color}{green}}
\def\bY{\begin{color}{yellow}}
\def\e{\end{color}\xspace}
\newcommand{\bit}{\begin{itemize}}
\newcommand{\eit}{\end{itemize}\par\noindent}
\newcommand{\ben}{\begin{enumerate}}
\newcommand{\een}{\end{enumerate}\par\noindent}
\newcommand{\beq}{\begin{equation}}
\newcommand{\eeq}{\end{equation}\par\noindent}
\newcommand{\beqa}{\begin{eqnarray*}}
\newcommand{\eeqa}{\end{eqnarray*}\par\noindent}
\newcommand{\beqn}{\begin{eqnarray}}
\newcommand{\eeqn}{\end{eqnarray}\par\noindent}

\usepackage{iftex}
\usepackage{comment}

\ifpdf
  \usepackage{underscore}         
  \usepackage[T1]{fontenc}        
  \usepackage[utf8]{inputenc}
  \usepackage[activate=true,
    final,
    tracking=true,
    factor=1100,
    stretch=10,
    shrink=10]{microtype}
    \microtypecontext{spacing=nonfrench}
\else
  \usepackage[hyphenbreaks]{breakurl}           
\fi

\title{Catalysing Completeness and Universality}

\author{Aleks Kissinger
\institute{University of Oxford}
\email{aleks.kissinger@ox.ac.uk} \and
Neil J. Ross
\institute{Dalhousie University}
\email{neil.jr.ross@dal.ca} \and
John van de Wetering
\institute{University of Amsterdam}
\email{john@vdwetering.name}
}

\def\titlerunning{Catalysing Completeness and Universality}
\def\authorrunning{A.~Kissinger, N.~J.~Ross \& J. van de Wetering}
\date{\today}

\begin{document}
\maketitle

\begin{abstract}
    A catalysis state is a quantum state that is used to make some desired operation possible or more efficient, while not being consumed in the process. Recent years have seen catalysis used in state-of-the-art protocols for implementing magic state distillation or small angle phase rotations. In this paper we will see that we can also use catalysis to prove that certain gate sets are computationally universal, and to extend completeness results of graphical languages to larger fragments. In particular, we give a simple proof of the computational universality of the CS+Hadamard gate set using the catalysis of a $T$ gate using a CS gate, which sidesteps the more complicated analytic arguments of the original proof by Kitaev. This then also gives us a simple self-contained proof of the computational universality of Toffoli+Hadamard. Additionally, we show that the phase-free ZH-calculus can be extended to a larger complete fragment, just by using a single catalysis rule (and one scalar rule).
\end{abstract}

In chemistry, a \emph{catalyst} is a substance that facilitates a
reaction without being modified by it. In analogy, a quantum state is
said to be a \emph{catalyst} or a \emph{catalysis state}, if it can be
used to make a desired operation more efficient, or even possible,
while not being consumed in the process.

The idea of catalysis has recently found several applications in
quantum computation. Catalytic methods were used to reason about resource
conversions in fault-tolerant quantum computing, and to derive lower
bounds on the cost of certain important computational
tasks~\cite{beverland2020}. Catalytic methods were also used to
improve the cost of unitary approximations over restricted gate sets
such as the Clifford+$T$ gate set~\cite{amy2023a}, and to establish
number-theoretic characterizations for important extensions of the
Clifford gate set~\cite{amy2023b}. A construction based on catalysis
is also currently the leading candidate for the fault-tolerant
implementation of small-angle
rotations~\cite{Gidney2018halvingcostof}, and one of the most promising 
magic state distillation protocols~\cite{Gidney2019efficientmagicstate}.

In this paper we extend the uses of catalysis in quantum computing in
two new directions: establishing the computational universality of
gate sets and proving the completeness of graphical calculi.

In the first direction we provide a novel proof of the computational
universality of the CS+Hadamard gate set by leveraging the fact that
any Clifford+$T$ circuit can be reduced to a CS+Hadamard circuit using
CS gates to catalyse $T$ gates. This side-steps the complicated
analytical arguments of the original proof~\cite{kitaev1997quantum},
and also establishes that CS+Hadamard circuits only incurs a linear
overhead in the number of samples needed compared to Clifford+$T$
circuits. In addition, this also leads to a simple and self-contained
proof of the computational universality of Toffoli+Hadamard circuits.

In the second direction we show that the phase-free
ZH-calculus~\cite{zhcompleteness2020,zhphasefree} can be extended to a
complete calculus for the Clifford+$T$ fragment by simply including a
rule encoding a catalytic equation, along with a scalar cancellation
rule. The proof of completeness then follows simply, and this approach
works for any tower of quadratic ring extensions
\[
\mathbb{D} \subseteq \mathbb{D}\left[a_1\right] \subseteq \ldots \subseteq \mathbb{D}\left[a_1, \ldots, a_k\right]
\]
where $\mathbb{D}=\mathbb{Z}[\frac12 ]$ and $a_j^2\in
\mathbb{D}[a_1,\ldots, a_{j-1}]$. In contrast to previous complete
calculi for the Clifford+$T$
fragment~\cite{SimonCompleteness,vilmartzxtriangle}, this yields a
calculus with rules that are easy to interpret, and relies on a
generic proof strategy, rather than one specific to Clifford+$T$.

More generally, our results demonstrate that catalytic methods provide
powerful means to extend results between different gate sets and graphical calculi.

\section{The ZH-calculus and catalysis}

The ZH-calculus~\cite{backens2018zhcalculus,zhcompleteness2020} is a graphical language designed to reason more easily about quantum computing involving controlled unitaries than the earlier ZX-calculus~\cite{CD1,CD2}. ZH-diagrams are string diagrams built out of generators representing certain linear maps between qubits that can be composed together either horizontally, corresponding to regular composition of linear maps, and vertically, corresponding to tensor product. The two generators are \emph{Z-spiders} and \emph{H-boxes}.
These are represented by circles and squares respectively, and correspond to the following linear maps:
\begin{equation}\label{eq:Z-spider-def}
 \begin{tikzpicture}
	\begin{pgfonlayer}{nodelayer}
		\node [style=none] (1) at (-1.5, 1) {};
		\node [style=none] (2) at (-1.5, -1) {};
		\node [style=none, rotate=90] (6) at (-1.25, 0) {...};
		\node [style=Z dot] (7) at (0, 0) {};
		\node [style=none] (8) at (1.5, 1) {};
		\node [style=none] (9) at (1.5, -1) {};
		\node [style=none, rotate=90] (11) at (1.25, 0) {...};
	\end{pgfonlayer}
	\begin{pgfonlayer}{edgelayer}
		\draw [in=180, out=45] (7) to (8.center);
		\draw [in=180, out=-45] (7) to (9.center);
		\draw [in=135, out=0] (1.center) to (7);
		\draw [in=225, out=0] (2.center) to (7);
	\end{pgfonlayer}
\end{tikzpicture} \ \ \ := \ \ \ket{0\cdots 0}\bra{0\cdots 0} + \ket{1\cdots 1}\bra{1\cdots 1} 
\end{equation}
\begin{equation}\label{eq:H-spider-def}
 \begin{tikzpicture}[decoration=brace]
	\begin{pgfonlayer}{nodelayer}
		\node [style=none] (0) at (1.25, -0.75) {};
		\node [style=none] (2) at (1.25, 0.75) {};
		\node [style=hadamard] (3) at (0, 0) {$a$};
		\node [style=none] (4) at (-1.25, -0.75) {};
		\node [style=none] (6) at (-1.25, 0.75) {};
		\node [style=none] (7) at (1.5, 1) {};
		\node [style=none] (8) at (1.5, -1) {};
		\node [style=none] (9) at (-1.5, 1) {};
		\node [style=none] (10) at (-1.5, -1) {};
		\node [style=none] (11) at (2.25, 0) {$n$};
		\node [style=none] (12) at (-2.25, 0) {$m$};
		\node [style=none, rotate=90] (13) at (1, 0) {$\ldots$};
		\node [style=none, rotate=90] (14) at (-1, 0) {$\ldots$};
	\end{pgfonlayer}
	\begin{pgfonlayer}{edgelayer}
		\draw [bend right=15] (3) to (6.center);
		\draw [bend left=15] (3) to (4.center);
		\draw [bend left=15] (0.center) to (3);
		\draw [bend left=15] (3) to (2.center);
		\draw [decorate] (7.center) to (8.center);
		\draw [decorate] (10.center) to (9.center);
	\end{pgfonlayer}
\end{tikzpicture} \ \ \ := \ \ \sum a^{i_1\ldots i_m j_1\ldots j_n} \ket{j_1\ldots j_n}\bra{i_1\ldots i_m}
\end{equation}
Here the \emph{label} of the H-box $a$ can be any complex number. Both Z-spiders and H-boxes can have any number of inputs or outputs (including zero). If they have $n$ inputs and $m$ outputs, then they correspond to matrices of size $2^n\times 2^m$. The Z-spider matrix consists of all zeroes, except for the top-left and bottom-right corner where there is a 1. The matrix of the H-box is all ones, except for the bottom-right corner where there is an $a$.
If the label of the H-box is $-1$, then we usually don't write it. In the special case of 1 input and 1 output the $-1$ labelled H-box is proportional to the Hadamard.
Next to these generators we also have the standard structural generators of compact-closed string-diagrammatic language: identity, swap, cup and cap~\cite{vandewetering2020zxcalculus}.

We say a graphical calculus is \emph{universal} for a set of matrices when it can represent any matrix in this set using some diagram. When we allow the label of H-boxes to be arbitrary complex numbers, ZH-diagrams are universal for all complex-valued matrices of size $2^n\times 2^m$~\cite{backens2018zhcalculus}. If instead we restrict the labels to some sub-ring $R$ of $\mathbb{C}$ including at least $\mathbb{Z}[\frac12]$, then it is universal for matrices over $R$~\cite{zhcompleteness2020}. In the \emph{phase-free} ZH-calculus we only allow the default label of $-1$ for the H-boxes, and we augment the calculus with a generator representing the scalar $\frac12$ written as a star: \begin{tikzpicture}
	\begin{pgfonlayer}{nodelayer}
		\node [style=st] (3) at (0, 0) {};
	\end{pgfonlayer}
\end{tikzpicture}
. The phase-free ZH-calculus is universal for matrices over the ring $\mathbb{Z}[\frac12]$~\cite{zhcompleteness2020}.

Let's give some examples of how some useful unitaries are represented as ZH-diagrams. First, We define an \emph{X-spider} and one with a \emph{NOT} applied to it as a derived generator:
\begin{equation}\label{eq:defx}
(X) \quad\ \  \begin{tikzpicture}[decoration=brace]
	\begin{pgfonlayer}{nodelayer}
		\node [style=gray dot] (0) at (-1.75, 0) {};
		\node [style=none] (1) at (-2.5, 1) {};
		\node [style=none] (2) at (-1, 1) {};
		\node [style=none] (3) at (-2.5, -1) {};
		\node [style=none] (4) at (-1, -1) {};
		\node [style=none] (5) at (0, 0) {$:=$};
		\node [style=white dot] (6) at (3.25, 0) {};
		\node [style=small hadamard] (7) at (2.25, 1) {};
		\node [style=small hadamard] (8) at (4.25, 1) {};
		\node [style=small hadamard] (9) at (2.25, -1) {};
		\node [style=small hadamard] (10) at (4.25, -1) {};
		\node [style=none] (11) at (2.25, 1.5) {};
		\node [style=none] (12) at (4.25, 1.5) {};
		\node [style=none] (13) at (2.25, -1.5) {};
		\node [style=none] (14) at (4.25, -1.5) {};
		\node [style=none] (15) at (3.25, 1) {$\ldots $};
		\node [style=none] (16) at (3.25, -1) {$\ldots $};
		\node [style=none] (17) at (-1.75, 0.75) {$\ldots $};
		\node [style=none] (18) at (-1.75, -0.75) {$\ldots $};
		\node [style=st] (31) at (1.5, 0) {};
	\end{pgfonlayer}
	\begin{pgfonlayer}{edgelayer}
		\draw [bend right=15] (1.center) to (0);
		\draw [bend right=15] (0) to (3.center);
		\draw [bend left=15] (2.center) to (0);
		\draw [bend left=15] (0) to (4.center);
		\draw (11.center) to (7);
		\draw [bend right=15] (7) to (6);
		\draw [bend right=15] (6) to (8);
		\draw (8) to (12.center);
		\draw [bend right=15] (6) to (9);
		\draw (9) to (13.center);
		\draw (10) to (14.center);
		\draw [bend left=15] (6) to (10);
	\end{pgfonlayer}
\end{tikzpicture}\qquad\qquad\qquad (NOT)\quad\ \  \begin{tikzpicture}
	\begin{pgfonlayer}{nodelayer}
		\node [style=gray dot] (0) at (-2, 0) {$\neg$};
		\node [style=none] (3) at (0, 0) {$:=$};
		\node [style=small hadamard] (6) at (4.05, 0.35) {};
		\node [style=small hadamard] (7) at (4.475, 0.95) {};
		\node [style=none] (10) at (-3, -1) {};
		\node [style=none] (11) at (-1, -1) {};
		\node [style=none] (12) at (-3, 1) {};
		\node [style=none] (13) at (-1, 1) {};
		\node [style=gray dot] (14) at (3, 0) {};
		\node [style=none] (15) at (2, -1) {};
		\node [style=none] (16) at (4, -1) {};
		\node [style=none] (17) at (2, 1) {};
		\node [style=none] (18) at (4, 1) {};
		\node [style=none] (19) at (3.05, 0.75) {$\cdots$};
		\node [style=none] (20) at (3.05, -0.75) {$\cdots$};
		\node [style=none] (21) at (-1.95, 0.75) {$\cdots$};
		\node [style=none] (22) at (-1.95, -0.75) {$\cdots$};
		\node [style=st] (23) at (1.5, 0) {};
	\end{pgfonlayer}
	\begin{pgfonlayer}{edgelayer}
		\draw [in=90, out=-90, looseness=0.75] (12.center) to (11.center);
		\draw [in=-90, out=90, looseness=0.75] (10.center) to (13.center);
		\draw [in=90, out=-90, looseness=0.75] (17.center) to (16.center);
		\draw [in=-90, out=90, looseness=0.75] (15.center) to (18.center);
		\draw [in=-135, out=15, looseness=0.75] (14) to (6);
		\draw [in=-105, out=45] (6) to (7);
	\end{pgfonlayer}
\end{tikzpicture}
\end{equation}
The X-spider allows us to calculate the XORs of computational basis states. Hence, in particular we can use it to represent a CX (i.e.~the CNOT) gate:
\begin{equation}
    \text{CNOT} \ = \ \begin{tikzpicture}
	\begin{pgfonlayer}{nodelayer}
		\node [style=cnot targ] (0) at (0, -0.5) {};
		\node [style=none] (1) at (-1, -0.5) {};
		\node [style=none] (2) at (1, -0.5) {};
		\node [style=cnot ctrl] (3) at (0, 0.5) {};
		\node [style=none] (4) at (-1, 0.5) {};
		\node [style=none] (5) at (1, 0.5) {};
	\end{pgfonlayer}
	\begin{pgfonlayer}{edgelayer}
		\draw (1.center) to (2.center);
		\draw (4.center) to (5.center);
		\draw (3) to (0);
	\end{pgfonlayer}
\end{tikzpicture} \ = \begin{tikzpicture}
	\begin{pgfonlayer}{nodelayer}
		\node [style=X dot] (56) at (0, -0.5) {};
		\node [style=Z dot] (57) at (0, 0.5) {};
		\node [style=none] (58) at (-1.25, -0.5) {};
		\node [style=none] (59) at (1.25, -0.5) {};
		\node [style=none] (60) at (-1.25, 0.5) {};
		\node [style=none] (61) at (1.25, 0.5) {};
	\end{pgfonlayer}
	\begin{pgfonlayer}{edgelayer}
		\draw (56) to (57);
		\draw (58.center) to (59.center);
		\draw (60.center) to (61.center);
	\end{pgfonlayer}
\end{tikzpicture}
\end{equation}
Here we are allowed to write a horizontal wire, because all the (derived) generators of the ZH-calculus are fully symmetric tensors, and hence whether a wire is an input or output does not change the linear map it represents.

Other useful gates are the CZ, CCZ and Toffoli gate:
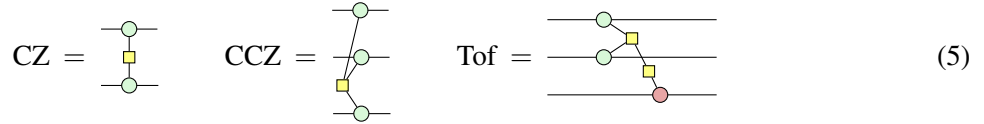
\begin{equation}
    \text{CZ} \ = \ \begin{tikzpicture}
	\begin{pgfonlayer}{nodelayer}
		\node [style=none] (31) at (-0.725, 0.75) {};
		\node [style=none] (32) at (-0.725, -0.75) {};
		\node [style=white dot] (33) at (0.025, -0.75) {};
		\node [style=white dot] (34) at (0.025, 0.75) {};
		\node [style=hadamard] (35) at (0.025, 0) {};
		\node [style=none] (37) at (0.8, -0.75) {};
		\node [style=none] (38) at (0.775, 0.75) {};
	\end{pgfonlayer}
	\begin{pgfonlayer}{edgelayer}
		\draw (37.center) to (32.center);
		\draw (34) to (33);
		\draw (31.center) to (38.center);
	\end{pgfonlayer}
\end{tikzpicture} \qquad \text{CCZ} \ = \ \begin{tikzpicture}
	\begin{pgfonlayer}{nodelayer}
		\node [style=none] (31) at (-0.725, 0) {};
		\node [style=none] (32) at (-0.725, -1.5) {};
		\node [style=white dot] (33) at (0.025, -1.5) {};
		\node [style=white dot] (34) at (0.025, 0) {};
		\node [style=hadamard] (35) at (-0.475, -0.75) {};
		\node [style=none] (37) at (0.8, -1.5) {};
		\node [style=none] (38) at (0.775, 0) {};
		\node [style=none] (39) at (-0.75, 1.25) {};
		\node [style=none] (40) at (0.75, 1.25) {};
		\node [style=white dot] (41) at (0, 1.25) {};
	\end{pgfonlayer}
	\begin{pgfonlayer}{edgelayer}
		\draw (37.center) to (32.center);
		\draw (31.center) to (38.center);
		\draw (41) to (35);
		\draw (34) to (35);
		\draw (33) to (35);
		\draw (40.center) to (39.center);
	\end{pgfonlayer}
\end{tikzpicture} \qquad \text{Tof} \ = \ \begin{tikzpicture}[decoration=brace]
	\begin{pgfonlayer}{nodelayer}
		\node [style=none] (9) at (-2.25, 1) {};
		\node [style=none] (10) at (-2.25, 0) {};
		\node [style=none] (11) at (-2.25, -1) {};
		\node [style=none] (12) at (2.25, 1) {};
		\node [style=none] (13) at (2.25, 0) {};
		\node [style=none] (14) at (2.25, -1) {};
		\node [style=white dot] (15) at (-0.75, 1) {};
		\node [style=white dot] (16) at (-0.75, 0) {};
		\node [style=gray dot] (17) at (0.75, -1) {};
		\node [style=hadamard] (18) at (0, 0.5) {};
		\node [style=hadamard] (19) at (0.45, -0.375) {};
	\end{pgfonlayer}
	\begin{pgfonlayer}{edgelayer}
		\draw (9.center) to (12.center);
		\draw (13.center) to (10.center);
		\draw (11.center) to (14.center);
		\draw (15) to (18);
		\draw (18) to (16);
		\draw (18) to (19);
		\draw (19) to (17);
	\end{pgfonlayer}
\end{tikzpicture}
\end{equation}
More general controlled-phase gates, can also be represented as ZH-diagrams. In particular, the CZ$(\alpha) := \text{diag}(1,1,1,e^{i\alpha})$ gate is represented as follows:
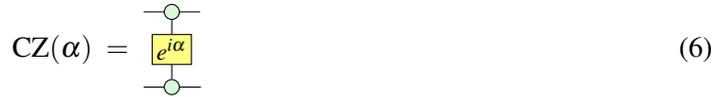
\begin{equation}
    \text{CZ}(\alpha) \ = \ \begin{tikzpicture}
	\begin{pgfonlayer}{nodelayer}
		\node [style=none] (31) at (-0.725, 1) {};
		\node [style=none] (32) at (-0.725, -1) {};
		\node [style=white dot] (33) at (0.025, -1) {};
		\node [style=white dot] (34) at (0.025, 1) {};
		\node [style=hadamard] (35) at (0.025, 0) {$e^{i\alpha}$};
		\node [style=none] (37) at (0.8, -1) {};
		\node [style=none] (38) at (0.775, 1) {};
	\end{pgfonlayer}
	\begin{pgfonlayer}{edgelayer}
		\draw (37.center) to (32.center);
		\draw (34) to (33);
		\draw (31.center) to (38.center);
	\end{pgfonlayer}
\end{tikzpicture}
\end{equation}

The ZH-\emph{calculus} is called that because we can actually do calculations with the diagrams. We can treat ZH-diagrams as undirected graphs, because Z-spiders and H-boxes are fully symmetric tensors, and hence the only relevant information in the diagram is which generator is connected to which other. Besides these topological symmetries, we also have a set of \emph{rewrite rules}. We present here the rules for the phase-free ZH-calculus; see Figure~\ref{fig:phasefree-rules}.


\begin{figure}[!tb]
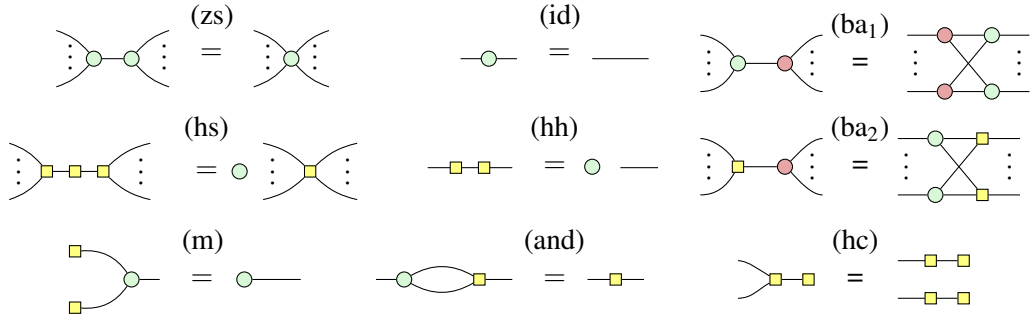

 \ctikzfig{ZH-rules} 
 \caption{The rules of the phase-free ZH-calculus.
 The right-hand sides of both \textit{bialgebra} rules \StrongCompRule and \HCompRule are complete bipartite graphs, with an additional input or output for each vertex. \SpiderRule and \HFuseRule stand respectively for \emph{Z-spider} and \emph{H-spider}; 
 \IdRule for \emph{identity}; 
 \MultRule for \emph{multiply}; 
 \AndRule for the identity involving the AND stating AND$(x,x) = x$; 
 and \HCopyRule for H-copy.
 }
 \label{fig:phasefree-rules}
\end{figure}
  
We say a set of rewrite rules is \emph{complete} for a ring $R$ when any two diagrams representing equal linear maps with matrix entries in the ring $R$ can be rewritten into each other using just these rewrite rules. 
The rule set of Figure~\ref{fig:phasefree-rules} is complete for the ring $\Z[\frac12]$.

Note that we have the following relations between the spiders of the ZX-calculus~\cite{CD1,CD2} and H-boxes:
\begin{equation}\label{eq:Z-a-state}
    \begin{tikzpicture}
	\begin{pgfonlayer}{nodelayer}
		\node [style=white phase dot] (0) at (0, 0) {$\alpha$};
		\node [style=none] (2) at (1.25, 0) {};
	\end{pgfonlayer}
	\begin{pgfonlayer}{edgelayer}
		\draw (0) to (2.center);
	\end{pgfonlayer}
\end{tikzpicture}
 \ = \ \begin{tikzpicture}
	\begin{pgfonlayer}{nodelayer}
		\node [style=hadamard] (0) at (0, 0) {$e^{i\alpha}$};
		\node [style=none] (2) at (1.5, 0) {};
	\end{pgfonlayer}
	\begin{pgfonlayer}{edgelayer}
		\draw (0) to (2.center);
	\end{pgfonlayer}
\end{tikzpicture}
 \qquad \begin{tikzpicture}
	\begin{pgfonlayer}{nodelayer}
		\node [style=hadamard] (0) at (-1.5, 0) {$1$};
		\node [style=none] (1) at (0, 0) {};
	\end{pgfonlayer}
	\begin{pgfonlayer}{edgelayer}
		\draw (0) to (1.center);
	\end{pgfonlayer}
\end{tikzpicture}
 \ = \ \begin{tikzpicture}
	\begin{pgfonlayer}{nodelayer}
		\node [style=Z dot] (0) at (-1.25, 0) {};
		\node [style=none] (1) at (0, 0) {};
	\end{pgfonlayer}
	\begin{pgfonlayer}{edgelayer}
		\draw (0) to (1.center);
	\end{pgfonlayer}
\end{tikzpicture}

\end{equation}
The following \emph{multiply rule} from the universal ZH-calculus (where H-boxes are labelled by arbitrary complex numbers) will also be useful:
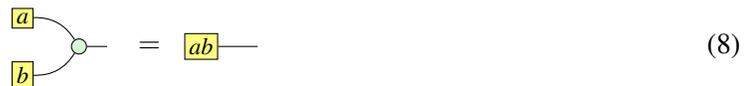
\begin{equation}\label{eq:multiply-rule-phased}
    \begin{tikzpicture}
	\begin{pgfonlayer}{nodelayer}
		\node [style=hadamard] (0) at (-3.375, 0.75) {$a$};
		\node [style=none] (1) at (-1.125, 0) {};
		\node [style=hadamard] (2) at (-3.375, -0.75) {$b$};
		\node [style=white dot] (3) at (-1.875, 0) {};
		\node [style=none] (4) at (0, 0) {$=$};
		\node [style=none] (5) at (2.875, 0) {};
		\node [style=hadamard] (6) at (1.375, 0) {$ab$};
	\end{pgfonlayer}
	\begin{pgfonlayer}{edgelayer}
		\draw (1.center) to (3);
		\draw [bend right] (3) to (0);
		\draw [bend left] (3) to (2);
		\draw (6) to (5.center);
	\end{pgfonlayer}
\end{tikzpicture}
\end{equation}
It is this equation that gives the multiply rule \MultRule it's name (take $a=b=-1$).

We can use the ZH-calculus to show the correctness of some simple catalysis equations.
\begin{example}
  \label{ex:cstot}
  We can use a single copy
  of a CS gate and a single $\ket{T}$ state to apply a $T$ gate and
  get the starting $\ket{T}$ state back:
  \begin{equation}\label{eq:cat-CS-T}
    \begin{tikzpicture}
	\begin{pgfonlayer}{nodelayer}
		\node [style=none] (0) at (-3.75, 1) {};
		\node [style=Z phase dot] (1) at (-3.5, -1) {$\frac\pi4$};
		\node [style=none] (2) at (-1, -1) {};
		\node [style=none] (3) at (-1, 1) {};
		\node [style=Z dot] (4) at (-1.75, 1) {};
		\node [style=Z dot] (5) at (-1.75, -1) {};
		\node [style=hadamard] (6) at (-1.75, 0) {$i$};
		\node [style=Z dot] (7) at (-2.5, 1) {};
		\node [style=X dot] (8) at (-2.5, -1) {};
		\node [style=none] (9) at (0, 0) {=};
		\node [style=none] (10) at (1, 1) {};
		\node [style=Z phase dot] (11) at (1.25, -1) {$\frac\pi4$};
		\node [style=none] (12) at (7, -1) {};
		\node [style=none] (13) at (7, 1) {};
		\node [style=Z dot] (17) at (2.25, 1) {};
		\node [style=X dot] (18) at (2.25, -1) {};
		\node [style=Z phase dot] (19) at (6, 1) {$\frac\pi4$};
		\node [style=Z phase dot] (20) at (6, -1) {$\frac\pi4$};
		\node [style=Z dot] (21) at (3, 1) {};
		\node [style=X dot] (22) at (3, -1) {};
		\node [style=Z dot] (23) at (5, 1) {};
		\node [style=X dot] (24) at (5, -1) {};
		\node [style=Z phase dot] (25) at (4, -1) {$\ -\frac\pi4~$};
		\node [style=none] (26) at (8, 0) {=};
		\node [style=none] (27) at (9, 1) {};
		\node [style=Z phase dot] (28) at (9.25, -1) {$\frac\pi4$};
		\node [style=none] (29) at (13.5, -1) {};
		\node [style=none] (30) at (13.5, 1) {};
		\node [style=Z phase dot] (33) at (12.5, 1) {$\frac\pi4$};
		\node [style=Z phase dot] (34) at (12.5, -1) {$\frac\pi4$};
		\node [style=Z dot] (37) at (11.5, 1) {};
		\node [style=X dot] (38) at (11.5, -1) {};
		\node [style=Z phase dot] (39) at (10.5, -1) {$\ -\frac\pi4~$};
		\node [style=none] (40) at (14.5, 0) {=};
		\node [style=none] (41) at (15.5, 1) {};
		\node [style=Z dot] (42) at (15.5, -1) {};
		\node [style=none] (43) at (18.5, -1) {};
		\node [style=none] (44) at (18.5, 1) {};
		\node [style=Z phase dot] (45) at (17.5, 1) {$\frac\pi4$};
		\node [style=Z phase dot] (46) at (17.5, -1) {$\frac\pi4$};
		\node [style=Z dot] (47) at (16.5, 1) {};
		\node [style=X dot] (48) at (16.5, -1) {};
		\node [style=none] (49) at (19.5, 0) {=};
		\node [style=none] (52) at (22.25, -1) {};
		\node [style=none] (53) at (22.25, 1) {};
		\node [style=Z phase dot] (54) at (21.25, 1) {$\frac\pi4$};
		\node [style=Z phase dot] (55) at (21.25, -1) {$\frac\pi4$};
		\node [style=none] (56) at (20.25, 1) {};
	\end{pgfonlayer}
	\begin{pgfonlayer}{edgelayer}
		\draw (2.center) to (1);
		\draw (0.center) to (3.center);
		\draw (4) to (6);
		\draw (6) to (5);
		\draw (8) to (7);
		\draw (12.center) to (11);
		\draw (10.center) to (13.center);
		\draw (18) to (17);
		\draw (22) to (21);
		\draw (24) to (23);
		\draw (29.center) to (28);
		\draw (27.center) to (30.center);
		\draw (38) to (37);
		\draw (43.center) to (42);
		\draw (41.center) to (44.center);
		\draw (48) to (47);
		\draw (54) to (53.center);
		\draw (55) to (52.center);
		\draw (56.center) to (54);
	\end{pgfonlayer}
\end{tikzpicture}
  \end{equation}
  Here, we used the standard decomposition of a controlled-phase gate
  CZ$(\alpha)$ into CNOT gates and $Z(\pm \alpha/2)$ phase gates.
\end{example}
  
Just using Clifford operations and CS gates, it is not possible to
construct a $T$ gate. This is because the number $e^{i\frac\pi4}$
appearing in the $T$ gate is not part of the ring of entries generated
by the matrices of the Clifford and CS gates. However, with
Eq.~\eqref{eq:cat-CS-T} we see that as soon as we have just \emph{one}
$\ket{T}$ state available to us, we can use $CS$ gates to apply
\emph{any number} of $T$ gates. Indeed, because the catalysis state
$\ket{T}$ is not modified by the application of the circuit in
Eq.~\eqref{eq:cat-CS-T}, it can be weaved through an arbitrary
Clifford+$T$ circuit in order to replace every $T$ gate by a small
Clifford+CS circuit. This catalytic use of $\ket{T}$ states can be
contrasted to the process of $T$ gate injection. In the latter
process, $\ket{T}$ states are used to perform $T$ gates but can no
longer be accessed after the injection has taken place: they have been
consumed by the injection.

\begin{example}
  \label{ex:cztot}
We can do something similar with a CCZ gate: we can transform the magic state $\ket{\text{CCZ}} := \text{CCZ}\ket{+++}$ into 3 $\ket{T}$ states using Clifford operations and one $T$ gate:
\begin{equation}
    \begin{tikzpicture}
	\begin{pgfonlayer}{nodelayer}
		\node [style=none] (0) at (17.62, -0.005) {};
		\node [style=none] (1) at (17.62, 1.005) {};
		\node [style=none] (2) at (17.62, -1) {};
		\node [style=Z dot] (3) at (11.5, 1.005) {};
		\node [style=Z dot] (4) at (11.5, -0.005) {};
		\node [style=Z phase dot] (5) at (11.5, -1) {$\frac{\pi}{4}$};
		\node [style=X dot] (6) at (10.4, 0.625) {};
		\node [style=X dot] (7) at (10.38, -0.38) {};
		\node [style=X dot] (8) at (12, -0.505) {};
		\node [style=Z phase dot] (9) at (13, -0.505) {$\ -\frac\pi4~$};
		\node [style=Z phase dot] (10) at (9.25, -0.38) {$\frac\pi4$};
		\node [style=Z phase dot] (11) at (9.25, 0.625) {$\ -\frac\pi4~$};
		\node [style=X phase dot] (12) at (15, 1.005) {$\pi$};
		\node [style=X phase dot] (13) at (14.25, -0.005) {$\pi$};
		\node [style=Z dot] (14) at (14.25, -1) {};
		\node [style=Z dot] (15) at (15, -1) {};
		\node [style=Z dot] (16) at (15.75, 1.005) {};
		\node [style=Z dot] (17) at (15.75, -0.005) {};
		\node [style=Z dot] (18) at (15.75, -1) {};
		\node [style=X dot] (19) at (16.25, 0.525) {};
		\node [style=Z phase dot] (20) at (17.25, 0.525) {$\ -\frac\pi4~$};
		\node [style=none] (21) at (-1.26, -0.005) {};
		\node [style=none] (22) at (-1.26, 1.005) {};
		\node [style=none] (23) at (-1.26, -1) {};
		\node [style=X phase dot] (33) at (-4.13, 1.005) {$\pi$};
		\node [style=X phase dot] (34) at (-4.88, -0.005) {$\pi$};
		\node [style=Z dot] (35) at (-4.88, -1) {};
		\node [style=Z dot] (36) at (-4.13, -1) {};
		\node [style=Z dot] (37) at (-2.88, 1.005) {};
		\node [style=Z dot] (38) at (-2.88, -0.005) {};
		\node [style=Z dot] (39) at (-2.88, -1) {};
		\node [style=X dot] (40) at (-2.38, 0.5) {};
		\node [style=Z phase dot] (41) at (-1.38, 0.5) {$\ -\frac\pi4~$};
		\node [style=hadamard] (42) at (-7.75, 0) {};
		\node [style=Z dot] (43) at (-7, 1) {};
		\node [style=Z dot] (44) at (-7, 0) {};
		\node [style=Z dot] (45) at (-7, -1) {};
		\node [style=X phase dot] (46) at (-6, -1) {$\ -\frac\pi2~$};
		\node [style=none] (47) at (6.86, -0.005) {};
		\node [style=none] (48) at (6.86, 1.005) {};
		\node [style=none] (49) at (6.86, -1) {};
		\node [style=X phase dot] (50) at (4.49, 1.005) {$\pi$};
		\node [style=X phase dot] (51) at (3.74, -0.005) {$\pi$};
		\node [style=Z dot] (52) at (3.74, -1) {};
		\node [style=Z dot] (53) at (4.49, -1) {};
		\node [style=Z dot] (54) at (5.24, 1.005) {};
		\node [style=Z dot] (55) at (5.24, -0.005) {};
		\node [style=Z dot] (56) at (5.24, -1) {};
		\node [style=X dot] (57) at (5.74, 0.5) {};
		\node [style=Z phase dot] (58) at (6.74, 0.5) {$\ -\frac\pi4~$};
		\node [style=hadamard] (59) at (1.12, 0) {};
		\node [style=Z dot] (60) at (1.87, 1) {};
		\node [style=Z dot] (61) at (1.87, 0) {};
		\node [style=none] (64) at (0, 0) {=};
		\node [style=hadamard] (65) at (2.75, 0.5) {$-i$};
		\node [style=none] (66) at (7.86, 0) {=};
		\node [style=none] (67) at (2.005, -3.505) {};
		\node [style=none] (68) at (2.005, -2.495) {};
		\node [style=none] (69) at (2.005, -4.5) {};
		\node [style=Z dot] (70) at (-4.115, -2.495) {};
		\node [style=Z dot] (71) at (-4.115, -3.505) {};
		\node [style=Z phase dot] (72) at (-4.115, -4.5) {$\frac{\pi}{4}$};
		\node [style=X dot] (74) at (-5.235, -3.88) {};
		\node [style=Z phase dot] (76) at (-0.865, -3.505) {$\frac\pi4$};
		\node [style=Z phase dot] (77) at (-6.365, -3.88) {$\frac\pi4$};
		\node [style=Z phase dot] (78) at (-0.865, -2.5) {$\frac\pi4$};
		\node [style=X phase dot] (79) at (-2.115, -2.495) {$\pi$};
		\node [style=X phase dot] (80) at (-2.865, -3.505) {$\pi$};
		\node [style=Z dot] (81) at (-2.865, -4.5) {};
		\node [style=Z dot] (82) at (-2.115, -4.5) {};
		\node [style=Z dot] (83) at (0.135, -2.495) {};
		\node [style=Z dot] (84) at (0.135, -3.505) {};
		\node [style=Z dot] (85) at (0.135, -4.5) {};
		\node [style=X dot] (86) at (0.635, -2.975) {};
		\node [style=Z phase dot] (87) at (1.635, -2.975) {$\ -\frac\pi4~$};
		\node [style=none] (88) at (-7.755, -3.5) {=};
		\node [style=none] (89) at (9.88, -3.505) {};
		\node [style=none] (90) at (9.88, -2.495) {};
		\node [style=none] (91) at (9.88, -4.5) {};
		\node [style=X dot] (95) at (5.14, -3.505) {};
		\node [style=Z phase dot] (96) at (7.01, -3.505) {$\frac\pi4$};
		\node [style=Z phase dot] (97) at (4.01, -3.505) {$\frac\pi4$};
		\node [style=Z phase dot] (98) at (7.01, -2.5) {$\frac\pi4$};
		\node [style=Z dot] (102) at (6.01, -4.5) {};
		\node [style=Z dot] (103) at (8.01, -2.495) {};
		\node [style=Z dot] (104) at (8.01, -3.505) {};
		\node [style=Z dot] (105) at (8.01, -4.5) {};
		\node [style=X dot] (106) at (8.51, -2.975) {};
		\node [style=Z phase dot] (107) at (9.51, -2.975) {$\ -\frac\pi4~$};
		\node [style=none] (108) at (3.12, -3.5) {=};
		\node [style=Z phase dot] (109) at (7.03, -4.505) {$\frac\pi4$};
		\node [style=Z dot] (110) at (6, -3.5) {};
		\node [style=Z dot] (111) at (6, -2.5) {};
		\node [style=none] (112) at (14.015, -3.505) {};
		\node [style=none] (113) at (14.015, -2.495) {};
		\node [style=none] (114) at (14.015, -4.5) {};
		\node [style=Z phase dot] (116) at (12.645, -3.505) {$\frac\pi4$};
		\node [style=Z phase dot] (118) at (12.645, -2.5) {$\frac\pi4$};
		\node [style=none] (125) at (11.255, -3.5) {=};
		\node [style=Z phase dot] (126) at (12.665, -4.505) {$\frac\pi4$};
	\end{pgfonlayer}
	\begin{pgfonlayer}{edgelayer}
		\draw (0.center) to (17);
		\draw (1.center) to (16);
		\draw (2.center) to (18);
		\draw (3) to (6);
		\draw (3) to (7);
		\draw (3) to (12);
		\draw (4) to (7);
		\draw (4) to (8);
		\draw (4) to (13);
		\draw (5) to (6);
		\draw (5) to (7);
		\draw (5) to (8);
		\draw (5) to (14);
		\draw (6) to (11);
		\draw (7) to (10);
		\draw (8) to (9);
		\draw (12) to (15);
		\draw (12) to (16);
		\draw (13) to (14);
		\draw (13) to (17);
		\draw (14) to (15);
		\draw (15) to (18);
		\draw (16) to (19);
		\draw (17) to (19);
		\draw (18) to (19);
		\draw (19) to (20);
		\draw (21.center) to (38);
		\draw (22.center) to (37);
		\draw (23.center) to (39);
		\draw (33) to (36);
		\draw (33) to (37);
		\draw (34) to (35);
		\draw (34) to (38);
		\draw (35) to (36);
		\draw (36) to (39);
		\draw (37) to (40);
		\draw (38) to (40);
		\draw (39) to (40);
		\draw (40) to (41);
		\draw (43) to (42);
		\draw (42) to (44);
		\draw (45) to (42);
		\draw (45) to (46);
		\draw (46) to (35);
		\draw (44) to (34);
		\draw (43) to (33);
		\draw (47.center) to (55);
		\draw (48.center) to (54);
		\draw (49.center) to (56);
		\draw (50) to (53);
		\draw (50) to (54);
		\draw (51) to (52);
		\draw (51) to (55);
		\draw (52) to (53);
		\draw (53) to (56);
		\draw (54) to (57);
		\draw (55) to (57);
		\draw (56) to (57);
		\draw (57) to (58);
		\draw (60) to (59);
		\draw (59) to (61);
		\draw (61) to (51);
		\draw (60) to (50);
		\draw (60) to (65);
		\draw (65) to (61);
		\draw [in=180, out=-30] (59) to (52);
		\draw (67.center) to (84);
		\draw (68.center) to (83);
		\draw (69.center) to (85);
		\draw (70) to (74);
		\draw (70) to (79);
		\draw (71) to (74);
		\draw (71) to (80);
		\draw (72) to (74);
		\draw (72) to (81);
		\draw (74) to (77);
		\draw (79) to (82);
		\draw (79) to (83);
		\draw (80) to (81);
		\draw (80) to (84);
		\draw (81) to (82);
		\draw (82) to (85);
		\draw (83) to (86);
		\draw (84) to (86);
		\draw (85) to (86);
		\draw (86) to (87);
		\draw (89.center) to (104);
		\draw (90.center) to (103);
		\draw (91.center) to (105);
		\draw (95) to (97);
		\draw (102) to (105);
		\draw (103) to (106);
		\draw (104) to (106);
		\draw (105) to (106);
		\draw (106) to (107);
		\draw (111) to (103);
		\draw (104) to (110);
		\draw (102) to (105);
		\draw (95) to (102);
		\draw (95) to (110);
		\draw (95) to (111);
		\draw (118) to (113.center);
		\draw (112.center) to (116);
		\draw (126) to (114.center);
	\end{pgfonlayer}
\end{tikzpicture}
\end{equation}
Here we use some rewrites involving phase gadgets and pushing phases through H-boxes~\cite{vandewetering2020zxcalculus}.
\end{example}

So again, if we can perform CCZ and Clifford gates and have just a single $\ket{T}$ available, then we can inject as many $T$ gates as we want.

These examples can be captured using the notion of \emph{catalytic embeddings}, which provide a framework for
reasoning about certain aspects of catalysis in quantum circuits~\cite{amy2023a}.

\begin{definition}
\label{def:catemb}
Let $\mathcal{U}$ and $\mathcal{V}$ be two collection of unitaries. An
\emph{$m$-dimensional catalytic embedding} from $\mathcal{U}$ to
$\mathcal{V}$ is a pair $(\phi, \ket{c})$ of a function $f:\mathcal{U}
\to \mathcal{V}$ and a quantum state $\ket{c}\in\C^m$ such that
\[
\phi(U) \ket{\psi}\ket{c} = (U\ket{\psi})\ket{c}
\]
for any unitary $U\in\mathcal{U}$ and any quantum state
$\ket{\psi}$. We call the state $\ket{c}$ the \emph{catalyst} of
$(\phi,\ket{c})$.
\end{definition}

Definition~\ref{def:catemb} shows that the unitary $\phi(U)$ can be
used to apply the unitary $U$ to the state $\ket{\psi}$, provided that
the state $\ket{c}$ is available.  Intuitively, one should think of
the elements of $\mathcal{U}$ as unitaries that are ``hard'' to
implement and of the elements of $\mathcal{V}$ as unitaries that are
``easy'' to implement. With this in mind, the existence of a catalytic
embedding $(\phi,\ket{c}): \mathcal{U} \to \mathcal{V}$ implies that
the hard unitaries can be performed using only the easy ones in the
presence of the appropriate catalyst. 

As a nice demonstration of the utility of ZH in catalysis, we present in Appendix~\ref{app:adder} a new proof of one of the most useful results using catalysis: how the synthesis of small-angle phase rotations can be implemented in a fault-tolerant friendly manner, by decomposing it into a series of dyadic angle rotations $Z(\frac{2\pi}{2^n})$, each of which can be implemented using catalysis and an adder gadget~\cite{Gidney2018halvingcostof}. While the end result is not new, our proof of correctness is more elementary, relying only on low level reasoning about catalysis, and does not use special high-level properties concerning phase gradients and Fourier transforms.

\section{Catalysing universality}

There are several ways in which a gate set can be \emph{universal}. It is common to say that a gate set $G$ is \emph{approximately universal} if any unitary can be approximated up to arbitrary accuracy using circuits over $G$. Because the fundamental purpose of a quantum computer is to estimate the expectation value of some observable $\mathcal{O}$, being able to approximately compute such expectation values yields another, weaker, notion of universality: \emph{computational universality}. 

We start with some state $\ket{\psi}$, apply some unitary $U$ to it, perform some measurements, and finally post-process these measurements to get an estimate of $\mathcal{O}$. After many such runs we will get a close approximation of $\mathcal{O}$. Mathematically we can represent this as
\begin{equation}\label{eq:expectation-value}
    \langle \mathcal{O}\rangle \ = \ 
    \begin{tikzpicture}
	\begin{pgfonlayer}{nodelayer}
		\node [style=wide left state] (0) at (-1, 0) {$\psi$};
		\node [style=big gate] (1) at (0.75, 0) {$U$};
		\node [style=big gate] (2) at (2.25, 0) {$\mathcal{O}$};
		\node [style=big gate] (3) at (3.75, 0) {$U^\dagger$};
		\node [style=wide right state] (4) at (5.5, 0) {$\psi$};
		\node [style=none] (5) at (-0.75, 0.75) {};
		\node [style=none] (6) at (5.25, 0.75) {};
		\node [style=none] (7) at (-0.75, -0.75) {};
		\node [style=none] (8) at (5.25, -0.75) {};
		\node [style=none] (9) at (-0.25, 0.25) {$\vdots$};
		\node [style=none] (10) at (4.75, 0.25) {$\vdots$};
	\end{pgfonlayer}
	\begin{pgfonlayer}{edgelayer}
		\draw (5.center) to (6.center);
		\draw (7.center) to (8.center);
	\end{pgfonlayer}
\end{tikzpicture}
\end{equation}
However, when we are trying to estimate this observable, we don't have to do this with just a single quantum circuit we run over and over again. Instead we can have a collection of different quantum circuits $V_j$ (potentially acting on a different number of qubits), input states $\ket{\psi_j}$, and observables $\mathcal{O}_j$, such that taking a particular weighted average gets us the outcome we are after:

\begin{equation}\label{eq:expectation-value-sum}
    \langle \mathcal{O}\rangle \ = \ \sum_j \lambda_j \ \begin{tikzpicture}
	\begin{pgfonlayer}{nodelayer}
		\node [style=wide left state] (0) at (-3.25, 0) {$\psi_j$};
		\node [style=big gate] (1) at (-1.5, 0) {$V_j$};
		\node [style=big gate] (2) at (0, 0) {$\mathcal{O}_j$};
		\node [style=big gate] (3) at (1.5, 0) {$V_j^\dagger$};
		\node [style=wide right state] (4) at (3.25, 0) {$\psi_j$};
		\node [style=none] (5) at (-3, 0.75) {};
		\node [style=none] (6) at (3, 0.75) {};
		\node [style=none] (7) at (-3, -0.75) {};
		\node [style=none] (8) at (3, -0.75) {};
		\node [style=none] (9) at (-2.5, 0.25) {$\vdots$};
		\node [style=none] (10) at (2.5, 0.25) {$\vdots$};
	\end{pgfonlayer}
	\begin{pgfonlayer}{edgelayer}
		\draw (5.center) to (6.center);
		\draw (7.center) to (8.center);
	\end{pgfonlayer}
\end{tikzpicture}
    \ = \ \sum_j \lambda_j \langle \mathcal{O}\rangle_j,
\end{equation}
where here we define $\langle \mathcal{O}\rangle_j$ to be the expectation value of $\mathcal{O}_j$ with respect to $V_j$ and $\ket{\psi_j}$.
We then see that if we can estimate each of the $\langle \mathcal{O}\rangle_j$, then we can also estimate $\langle \mathcal{O}\rangle$ itself, by just summing our estimates like $\langle \mathcal{O}\rangle = \sum_j \lambda_j \langle \mathcal{O}\rangle_j$.

While we can reduce the calculation of an expectation value to the calculation of a sum of (potentially simpler to calculate) expectation values in this way, there is the important issue of the overhead in the number of samples needed. Generally, we want to determine an error budget for how close we want the estimate to be, and then that determines how many times we need to sample from the quantum computation. Since we are summing together different expectation values, we need to be careful that we aren't blowing up the error in the estimates. Suppose for instance that some $\lambda_k = 100$. Then a small error in our estimate of $\langle \mathcal{O}\rangle_j$ will blow up by a factor of a 100. On the other hand, if $\lambda_k = 1/100$, then any error will also be decreased by a factor of a 100, so that even a large error is not that important. The most efficient strategy is then to sample $\langle \mathcal{O}\rangle_j$ a number of times proportional to $\lvert \lambda_j\rvert$. The total overhead using this sum-of-expectations approach is then $\sum_j \lvert \lambda_j\rvert$ when comparing it to estimating $\langle \mathcal{O}\rangle$ directly with the original circuit .

These sum-of-expectations techniques are used in a variety of subfields in quantum computing. For instance, quasi-probability simulators use such a technique to write a quantum computation as an affine combination of easier to classically simulate quantum computations~\cite{bravyiImprovedClassicalSimulation2016,veitch2012negative}. They are also used in stochastic compilation in order to better suppress errors, for instance by randomly multiplying a state by a Pauli in order to get rid of systematic errors~\cite{wallman2016noise}, or by manipulating the order of Trotter terms when decomposing a Hamiltonian simulation~\cite{Cambell2019RandomCompiler}.
Here we will see that such a technique can also be used to argue for the computational universality of a gate set, by reducing from a more extensive gate set using catalysis.

\subsection{Clifford+CS}

Here we will show that we can reduce the calculation of an expectation value involving a Clifford+T circuit to one involving a collection of Clifford+CS circuits with some small overhead. This will prove that the Clifford+CS gate set is computationally universal, as the Clifford+T gate set is as well.

Suppose we have a Clifford+T circuit $C$ applied to the input state $\ket{\psi}$. Then we can transform $C$ into a circuit $C'$ containing just Clifford gates and CS gates using catalysis, so that $C'\ket{\psi}\ket{T} = C\ket{\psi} \ket{T}$. 
If we were trying to estimate the observable $\mathcal{O}$ we can then check that:
\begin{equation}\label{eq:expectation-value-T-cat}
    \begin{tikzpicture}
	\begin{pgfonlayer}{nodelayer}
		\node [style=wide left state] (0) at (-10, 0) {$\psi$};
		\node [style=big gate, minimum height=1.6cm] (1) at (-8.25, -0.625) {$C'$};
		\node [style=big gate] (2) at (-6.75, 0) {$\mathcal{O}$};
		\node [style=big gate, minimum height=1.6cm] (3) at (-4.75, -0.625) {$(C')^\dagger$};
		\node [style=wide right state] (4) at (-2.625, 0) {$\psi$};
		\node [style=none] (5) at (-9.75, 0.75) {};
		\node [style=none] (6) at (-2.875, 0.75) {};
		\node [style=none] (7) at (-9.75, -0.75) {};
		\node [style=none] (8) at (-2.875, -0.75) {};
		\node [style=none] (9) at (-9.25, 0.25) {$\vdots$};
		\node [style=none] (10) at (-3.375, 0.25) {$\vdots$};
		\node [style=Z phase dot] (11) at (-10, -1.75) {$\frac\pi4$};
		\node [style=Z phase dot] (12) at (-2.375, -1.75) {$\ -\frac\pi4~$};
		\node [style=none] (13) at (0, 0) {=};
		\node [style=wide left state] (14) at (1.75, 0) {$\psi$};
		\node [style=big gate] (15) at (3.5, 0) {$C$};
		\node [style=big gate] (16) at (5, 0) {$\mathcal{O}$};
		\node [style=big gate] (17) at (7, 0) {$C^\dagger$};
		\node [style=wide right state] (18) at (9.125, 0) {$\psi$};
		\node [style=none] (19) at (2, 0.75) {};
		\node [style=none] (20) at (8.875, 0.75) {};
		\node [style=none] (21) at (2, -0.75) {};
		\node [style=none] (22) at (8.875, -0.75) {};
		\node [style=none] (23) at (2.5, 0.25) {$\vdots$};
		\node [style=none] (24) at (8.375, 0.25) {$\vdots$};
		\node [style=Z phase dot] (25) at (3.5, -1.75) {$\frac\pi4$};
		\node [style=Z phase dot] (26) at (7, -1.75) {$\ -\frac\pi4~$};
		\node [style=none] (27) at (0, 1) {\eqref{eq:cat-CS-T}};
		\node [style=none] (28) at (-11.25, -1.75) {$\frac{1}{\sqrt{2}}$};
		\node [style=none] (29) at (-1, -1.75) {$\frac{1}{\sqrt{2}}$};
		\node [style=none] (30) at (2.25, -1.75) {$\frac12$};
		\node [style=none] (31) at (0, -4.5) {=};
		\node [style=wide left state] (32) at (1.75, -4.5) {$\psi$};
		\node [style=big gate] (33) at (3.5, -4.5) {$C$};
		\node [style=big gate] (34) at (5, -4.5) {$\mathcal{O}$};
		\node [style=big gate] (35) at (7, -4.5) {$C^\dagger$};
		\node [style=wide right state] (36) at (9.125, -4.5) {$\psi$};
		\node [style=none] (37) at (2, -3.75) {};
		\node [style=none] (38) at (8.875, -3.75) {};
		\node [style=none] (39) at (2, -5.25) {};
		\node [style=none] (40) at (8.875, -5.25) {};
		\node [style=none] (41) at (2.5, -4.25) {$\vdots$};
		\node [style=none] (42) at (8.375, -4.25) {$\vdots$};
		\node [style=none] (45) at (0, -3.5) {\spiderrule};
	\end{pgfonlayer}
	\begin{pgfonlayer}{edgelayer}
		\draw (5.center) to (6.center);
		\draw (7.center) to (8.center);
		\draw (11) to (12);
		\draw (19.center) to (20.center);
		\draw (21.center) to (22.center);
		\draw (25) to (26);
		\draw (37.center) to (38.center);
		\draw (39.center) to (40.center);
	\end{pgfonlayer}
\end{tikzpicture}
\end{equation}
So instead of running the circuit $C$, we can run the circuit $C'$, which doesn't contain any $T$ gates. This is then an example of Eq.~\eqref{eq:expectation-value-sum} where the sum is over just one term and we have $\lambda_1:=1$, $U_1 := C'$, $\ket{\psi_1} := \ket{\psi}\otimes \ket{T}$ and $\mathcal{O}_1 := \mathcal{O}\otimes I$. 

But to prepare $\ket{T}$ we still need to use a $T$ gate, so we need to also get rid of this magic state. We can decompose this magic state into a sum of Clifford states. Because each term in the sum needs to retain the form of an expectation value like~\eqref{eq:expectation-value}, we can't just decompose $\ket{T}$ into pure states $\ket{\phi_j}$, instead we need to decompose $\ketbra{T}{T}$ into a sum of density matrices $\ketbra{\phi_j}{\phi_j}$. One way to do this is the following:
\begin{align}\label{eq:ketbra-T-decomp}
    \ketbra{T}{T} \ &= \ \frac12\begin{pmatrix}1 & e^{i\pi/4} \\ e^{-i\pi/4} & 1\end{pmatrix} \ = \ \frac12\begin{pmatrix}1 & \frac{1+i}{\sqrt{2}} \\ \frac{1-i}{\sqrt{2}} & 1\end{pmatrix} \ \nonumber\\
    &= \ \frac{1}{\sqrt{2}} \left( \ketbra{+}{+} + \ketbra{-i}{-i}\right) - \frac{\sqrt{2}-1}{2}\left(\ketbra{0}{0} + \ketbra{1}{1}\right).
\end{align}
Hence, we can decompose $\ketbra{T}{T}$ into four Clifford states $\ket{\phi_1} = \ket{+}$, $\ket{\phi_2} = \ket{-i}$, $\ket{\phi_3} = \ket{0}$ and $\ket{\phi_4} = \ket{1}$ with weights $\lambda_1 = \lambda_2 = \frac{1}{\sqrt{2}}$ and $\lambda_3=\lambda_4 = - \frac{\sqrt{2}-1}{2}$. 
Starting with the left-hand side of Eq.~\eqref{eq:expectation-value-T-cat} we then have:
\begin{equation}
    \scalebox{0.9}{\begin{tikzpicture}
	\begin{pgfonlayer}{nodelayer}
		\node [style=wide left state] (0) at (-9.25, 0) {$\psi$};
		\node [style=big gate, minimum height=1.6cm] (1) at (-7.5, -0.625) {$C'$};
		\node [style=big gate] (2) at (-6, 0) {$\mathcal{O}$};
		\node [style=big gate, minimum height=1.6cm] (3) at (-4, -0.625) {$(C')^\dagger$};
		\node [style=wide right state] (4) at (-1.875, 0) {$\psi$};
		\node [style=none] (5) at (-9, 0.75) {};
		\node [style=none] (6) at (-2.125, 0.75) {};
		\node [style=none] (7) at (-9, -0.75) {};
		\node [style=none] (8) at (-2.125, -0.75) {};
		\node [style=none] (9) at (-8.5, 0.25) {$\vdots$};
		\node [style=none] (10) at (-2.625, 0.25) {$\vdots$};
		\node [style=Z phase dot] (11) at (-9.25, -1.75) {$\frac\pi4$};
		\node [style=Z phase dot] (12) at (-1.625, -1.75) {$\ -\frac\pi4~$};
		\node [style=none] (13) at (0, 0) {=};
		\node [style=none] (27) at (0, 1) {\eqref{eq:cat-CS-T}};
		\node [style=none] (28) at (-10.5, -1.75) {$\frac12$};
		\node [style=wide left state] (29) at (3, 0) {$\psi$};
		\node [style=big gate, minimum height=1.6cm] (30) at (4.75, -0.625) {$C'$};
		\node [style=big gate] (31) at (6.25, 0) {$\mathcal{O}$};
		\node [style=big gate, minimum height=1.6cm] (32) at (8.25, -0.625) {$(C')^\dagger$};
		\node [style=wide right state] (33) at (10.375, 0) {$\psi$};
		\node [style=none] (34) at (3.25, 0.75) {};
		\node [style=none] (35) at (10.125, 0.75) {};
		\node [style=none] (36) at (3.25, -0.75) {};
		\node [style=none] (37) at (10.125, -0.75) {};
		\node [style=none] (38) at (3.75, 0.25) {$\vdots$};
		\node [style=none] (39) at (9.625, 0.25) {$\vdots$};
		\node [style=Z phase dot] (40) at (3.5, -1.75) {$\frac\pi4$};
		\node [style=none] (42) at (1.5, 0) {$\frac12$};
		\node [style=Z phase dot] (43) at (2.375, -1.75) {$\ -\frac\pi4~$};
		\node [style=none] (44) at (1.625, -2.75) {};
		\node [style=none] (45) at (9.75, -2.75) {};
		\node [style=none] (46) at (9.75, -1.75) {};
		\node [style=none] (47) at (1.625, -1.75) {};
		\node [style=wide left state] (48) at (-9.125, -4.75) {$\psi$};
		\node [style=big gate, minimum height=1.6cm] (49) at (-7.375, -5.375) {$C'$};
		\node [style=big gate] (50) at (-5.875, -4.75) {$\mathcal{O}$};
		\node [style=big gate, minimum height=1.6cm] (51) at (-3.875, -5.375) {$(C')^\dagger$};
		\node [style=wide right state] (52) at (-1.75, -4.75) {$\psi$};
		\node [style=none] (53) at (-8.875, -4) {};
		\node [style=none] (54) at (-2, -4) {};
		\node [style=none] (55) at (-8.875, -5.5) {};
		\node [style=none] (56) at (-2, -5.5) {};
		\node [style=none] (57) at (-8.375, -4.5) {$\vdots$};
		\node [style=none] (58) at (-2.5, -4.5) {$\vdots$};
		\node [style=left state] (59) at (-8.625, -6.5) {};
		\node [style=right state] (61) at (-10.25, -6.5) {};
		\node [style=none] (62) at (-10.75, -7.5) {};
		\node [style=none] (63) at (-2.375, -7.5) {};
		\node [style=none] (64) at (-2.375, -6.5) {};
		\node [style=none] (65) at (-10.75, -6.5) {};
		\node [style=none] (66) at (-10.25, -6.5) {{\scriptsize $\phi_j$}};
		\node [style=none] (67) at (-8.625, -6.5) {{\scriptsize $\phi_j$}};
		\node [style=none] (68) at (-13.5, -4.75) {=};
		\node [style=none] (69) at (-13.5, -3.75) {\eqref{eq:ketbra-T-decomp}};
		\node [style=none] (70) at (-12.25, -6.5) {$\sum_j\lambda_j$};
		\node [style=none] (71) at (0, -4.75) {=};
		\node [style=wide left state] (72) at (4, -4.75) {$\psi$};
		\node [style=big gate, minimum height=1.6cm] (73) at (5.75, -5.375) {$C'$};
		\node [style=big gate] (74) at (7.25, -4.75) {$\mathcal{O}$};
		\node [style=big gate, minimum height=1.6cm] (75) at (9.25, -5.375) {$(C')^\dagger$};
		\node [style=wide right state] (76) at (11.375, -4.75) {$\psi$};
		\node [style=none] (77) at (4.25, -4) {};
		\node [style=none] (78) at (11.125, -4) {};
		\node [style=none] (79) at (4.25, -5.5) {};
		\node [style=none] (80) at (11.125, -5.5) {};
		\node [style=none] (81) at (4.75, -4.5) {$\vdots$};
		\node [style=none] (82) at (10.625, -4.5) {$\vdots$};
		\node [style=left state] (83) at (4, -6.5) {};
		\node [style=right state] (84) at (11.375, -6.5) {};
		\node [style=none] (89) at (11.375, -6.5) {{\scriptsize $\phi_j$}};
		\node [style=none] (90) at (4, -6.5) {{\scriptsize $\phi_j$}};
		\node [style=none] (91) at (1.875, -4.75) {$\sum_j\lambda_j$};
	\end{pgfonlayer}
	\begin{pgfonlayer}{edgelayer}
		\draw (5.center) to (6.center);
		\draw (7.center) to (8.center);
		\draw (11) to (12);
		\draw (34.center) to (35.center);
		\draw (36.center) to (37.center);
		\draw (44.center) to (45.center);
		\draw [bend left=90, looseness=1.50] (46.center) to (45.center);
		\draw (46.center) to (40);
		\draw [in=180, out=-180, looseness=1.50] (47.center) to (44.center);
		\draw (47.center) to (43);
		\draw (53.center) to (54.center);
		\draw (55.center) to (56.center);
		\draw (62.center) to (63.center);
		\draw [bend left=90, looseness=1.50] (64.center) to (63.center);
		\draw (64.center) to (59);
		\draw [in=180, out=-180, looseness=1.50] (65.center) to (62.center);
		\draw (65.center) to (61);
		\draw (77.center) to (78.center);
		\draw (79.center) to (80.center);
		\draw (89.center) to (90.center);
	\end{pgfonlayer}
\end{tikzpicture}}
\end{equation}

We see then that this is a case of Eq.~\eqref{eq:expectation-value-sum} with $\ket{\psi_j} := \ket{\psi}\otimes \ket{\phi_j}$ and $\mathcal{O}_j := \mathcal{O}\otimes I$ and $U_j = C'$ for all $j\in\{1,2,3,4\}$.
Furthermore, we can check that the four terms have $\sum_j \lvert \lambda_j \rvert = 2\sqrt{2} -1 \approx 1.83$.
Hence, if we decompose the magic state in this way we need to collect $1.83$ samples more than we would have needed to if we did use the magic $\ket{T}$ state directly. 

Summarising the full procedure we see then that we can do the following:
\begin{enumerate}
    \item Start with the Clifford+T computation you want to calculate.
    \item Replace all $T$ gates by a CS gate catalysis circuit using a $\ket{T}$.
    \item Replace the $\ket{T}$ state needed for all the catalysis by the Clifford states $\ket{\psi_j}$.
    \item Run each of the resulting four circuits a number of times proportional to $\lvert \lambda_j\rvert$.
    \item Combine the resulting estimates of the observable by scaling by $\lambda_j$ to get the final outcome.
\end{enumerate}
When we have Clifford gates and CS gates, the gate set is generated by \CNOT, Hadamard, $S$ and CS. The \CNOT can be constructed using CS and Hadamard, and if we allow states to be prepared in the $\ket{0}$ and $\ket{1}$ (which is necessary to encode different input states), then we can also prepare an $S$ using a CS. Hence, this gate set is equivalent to just the CS and Hadamard gate.
We see then that we have proven the following.
\begin{theorem}\label{thm:CS-Had-universal}
    The CS+Hadamard gate set is computationally universal. In particular, a Clifford+T computation can be simulated by a CS+Hadamard computation with a linear overhead in the number of samples, qubits and gates needed. 
\end{theorem}

\begin{remark}
    Without using the catalysis, we could have also chosen to write each $T$ gate as a magic state injection, and then replace each of the $\ketbra{T}{T}$ states by the Clifford states $\ketbra{\phi_j}{\phi_j}$. When we do this however, we get a number of terms in the decompositions that scales exponentially in the number of $T$ gates, and in particular $\sum_j \lvert\lambda_j\rvert$ scales exponentially, so that the simulation would no longer be efficient. This makes sense, since replacing all the $T$ gates gives us a Clifford circuit, and we don't expect this gate set to be computationally universal.
\end{remark}

\begin{remark}
    CS+Hadamard is in fact approximately universal, as proven by Kitaev~\cite[Lemma~4.6 on p.~1213]{kitaev1997quantum} (note that what he calls $S$ is the Hadamard gate, and $K$ is the $S$ gate). Kitaev proves this using an unproven `geometric lemma' that adding a gate to a set of gates that stabilises a given state creates a larger-dimensional space of gates. This proof is not constructive. He then proves what is now known as the Solovay-Kitaev algorithm to show how you would do it constructively. Though this result also establishes the computational universality of CS+Hadamard, we remark that that construction would, for instance, give an approximate decomposition of the $T$ gate, meaning the cost of implementing the $T$ would scale with the desired precision, whereas our method has a constant overhead.
\end{remark}

\subsection{Real-valued unitaries}

An interesting example of a computationally universal gate set is the set of real-valued unitaries~\cite{aharonov2003simple}. That is, where we only allow unitaries where all matrix entries are real numbers. Obviously this gate set is not approximately universal as we can never approximate any complex-valued unitary (like the $S$ gate), but it turns out we can `simulate' complex-valued unitaries using a real-valued one on a larger set of qubits.

For a (complex-valued) $n$-qubit unitary $U$, let $\Re(U)$ be the real part of $U$. That is: $\Re(U)_{ij} = \Re(U_{ij})$. Similarly define the complex part $\Im(U)$. Then $U = \Re(U) + i \Im(U)$.
Now define the $(n+1)$-qubit unitary $\tilde{U}$ via
\begin{align*}
    \tilde{U}(\ket{0}\otimes \ket{\psi}) \ &:= \ \ket{0}\otimes (\Re(U) \ket{\psi}) + \ket{1} \otimes (\Im(U)\ket{\psi}) \\
    \tilde{U}(\ket{1}\otimes \ket{\psi}) \ &:= \ - \ket{0}\otimes (\Im(U) \ket{\psi}) + \ket{1} \otimes (\Re(U)\ket{\psi})
\end{align*}

While it is clear that $\tilde{U}$ is real-valued, it is not immediately obvious that it is unitary. We show this, and the fact that $\widetilde{UV} = \tilde{U}\tilde{V}$ in Appendix~\ref{app:real-valued}.

This construction from~\cite{aharonov2003simple} can actually also be seen as an example of catalysis as $\tilde{U}\ket{-i}\otimes \ket{\psi} = \ket{-i} \otimes U\ket{\psi}$.
This then might look like we can use the previous sum-of-expectations approach to argue that real-valued unitaries are computationally universal. This does not work however, as we cannot write $\ketbra{-i}{-i}$ as a sum of self-adjoint real matrices, so that we can't represent it directly just using real-valued unitaries and state preparations.
We can however still prove computational universality using a slightly different argument.


In fact, we can exactly simulate the probability distribution arising from a complex-valued circuit. To see this, first note that if $\ket{\psi}$ and $\ket{\psi'}$ are real-valued states that then $\lvert \bra{\psi} C \ket{\psi'}\rvert^2 = \lvert \bra{\psi} \Re(C) \ket{\psi'}\rvert^2 + \lvert \bra{\psi} \Im(C) \ket{\psi'}\rvert^2$ (this requires some computation to show). On the other hand, if we input $\ket{0}\otimes \ket{\psi'}$ into $\tilde{C}$ and do a measurement marginalising over the first qubit we also get the probabilities:
\[
    \sum_{x=0,1} \lvert \bra{x, \psi} \tilde{C} \ket{0, \psi'} \rvert^2 \ = \ \lvert \bra{\psi} \Re(C) \ket{\psi'}\rvert^2 + \lvert \bra{\psi} \Im(C) \ket{\psi'}\rvert^2 \ = \ \lvert \bra{\psi} C \ket{\psi'}\rvert^2.
\]
So the probability distribution we get for $C$ is the same as that for $\tilde{C}$ when we prepare the first qubit in the $\ket{0}$ state and ignore its measurement outcome.

\begin{proposition}\label{prop:computationally-universal-real}
    Real-valued unitaries are computationally universal.
\end{proposition}

\subsection{Toffoli+Hadamard}

Because we can simulate complex-valued quantum circuits using real-valued unitaries in the direct manner described above, we don't need \emph{all} the real-valued unitaries. Given any computationally universal gate set $G$ we only need $\tilde{U}$ for $U\in G$. 

In particular, for CS+Hadamard, we can check that we get a real-valued unitary $\widetilde{\text{CS}}$ that is equivalent to a Toffoli up to some swaps. With the Hadamard we just get $\tilde{H} = I\otimes H$. Hence, when we encode the CS+Hadamard gate set, we get the Toffoli+Hadamard gate set~\cite{aharonov2003simple}.
We can hence do the following: starting with a Clifford+T computation, we write it as an ensemble of CS+Hadamard circuits. We then encode each of these circuits into a real-valued Toffoli+Hadamard circuit. By doing this we can efficiently simulate the original Clifford+T circuit. We see then that Toffoli+Hadamard circuits are also computationally universal.
\begin{theorem}
    The Toffoli+Hadamard gate set is computationally universal.
\end{theorem}

Note that we could also do a version of this without first encoding a complex unitary as a real unitary by using the fact that we can catalyse $T$ gates using a CCZ gate. We then get a version of Theorem~\ref{thm:CS-Had-universal} for the Clifford+CCZ gate set. Using this method, we however also need the $S$ gate, which is not necessary using the above approach.

Our proof of the universality of Toffoli+Hadamard follows along the same lines as that of Aharonov~\cite{aharonov2003simple}, namely, by reducing it to CS+Hadamard. However, the original proof of CS+Hadamard universality is non-constructive and relies on a series of non-trivial encodings of unitaries into larger-dimensional spaces, whilst our proof reduces the problem to Clifford+$T$ universality, which itself simply reduces to universality of CNOT+single-qubit unitaries.

\subsection{Catalysing general gate sets}

We can generalise these specific statements of computational universality to a general statement about gate sets that have catalytic embeddings. We present he proof in Appendix~\ref{app:proof-catalysis}.

\begin{theorem}\label{thm:catalysis-universality}
    Let $\mathcal{U}$ and $\mathcal{V}$ be gate sets which have a catalytic embedding $(\phi, \ket{c})$ as in Definition~\ref{def:catemb} and suppose that $\ketbra{c}{c}$ can be written as a sum $\sum_j\lambda_j \ketbra{\psi_j}{\psi_j}$ where each of the $\ket{\psi_i}$ can be prepared by a circuit over $\mathcal{V}$. Then if the bigger gate set $\mathcal{U}$ is computationally universal, the smaller gate set $\mathcal{V}$ is also computationally universal.
\end{theorem}

In this result we needed the catalyst $\ketbra{c}{c}$ to be expressible as a sum of states that can be prepared by the smaller gate set. In the examples of catalysis we saw earlier, these states were all stabiliser states, and hence could be expressed as the gate set included all Cliffords. There is however nothing special about a stabiliser decomposition, and any decomposition into expressible states is sufficient. Even an approximate decomposition would suffice, as long as the 1-norm $\sum_j \lvert \lambda_j \rvert$ of the approximate decomposition scales polynomially in the desired error rate $\varepsilon$. This condition is needed because the error in the catalysis state needs to be lower if it is used more often, and hence the 1-norm should not increase too rapidly.

\section{Catalysing Completeness}

We can also use catalysis in order to define extensions of graphical calculi and prove completeness for them. 
To see how this works, we want to first generalise the $T$ gate catalysis of Eq.~\eqref{eq:cat-CS-T}. First, as our goal will just be to produce states, we can plug $\ket{+}$ into the top wire of Eq.~\eqref{eq:cat-CS-T}. We can then simplify the expression to a more symmetric form:
\begin{equation}\label{eq:cat-CS-T-ZH}
    \begin{tikzpicture}
	\begin{pgfonlayer}{nodelayer}
		\node [style=Z dot] (0) at (5, 1) {};
		\node [style=Z phase dot] (1) at (5, -1) {$\frac\pi4$};
		\node [style=none] (2) at (7.5, -1) {};
		\node [style=none] (3) at (7.5, 1) {};
		\node [style=Z dot] (4) at (6.75, 1) {};
		\node [style=Z dot] (5) at (6.75, -1) {};
		\node [style=hadamard] (6) at (6.75, 0) {$i$};
		\node [style=Z dot] (7) at (6, 1) {};
		\node [style=X dot] (8) at (6, -1) {};
		\node [style=none] (9) at (0, 0) {=};
		\node [style=none] (43) at (-1.25, -1) {};
		\node [style=none] (44) at (-1.25, 1) {};
		\node [style=Z phase dot] (45) at (-2.25, 1) {$\frac\pi4$};
		\node [style=Z phase dot] (46) at (-2.25, -1) {$\frac\pi4$};
		\node [style=none] (48) at (2.75, -1) {};
		\node [style=none] (49) at (2.75, 1) {};
		\node [style=Z phase dot] (50) at (1.75, 1) {$\frac\pi4$};
		\node [style=Z phase dot] (51) at (1.75, -1) {$\frac\pi4$};
		\node [style=Z dot] (52) at (0.75, 1) {};
		\node [style=none] (53) at (3.75, 0) {=};
		\node [style=Z dot] (54) at (9.25, 1) {};
		\node [style=none] (56) at (13.25, -1) {};
		\node [style=none] (57) at (13.25, 1) {};
		\node [style=Z dot] (58) at (12.5, 1) {};
		\node [style=Z dot] (59) at (12.5, -1) {};
		\node [style=hadamard] (60) at (12.5, 0) {$i$};
		\node [style=Z dot] (61) at (10, 1) {};
		\node [style=X dot] (62) at (10, -1) {};
		\node [style=none] (63) at (8.5, 0) {=};
		\node [style=Z phase dot] (64) at (10.75, 0) {$\frac\pi4$};
		\node [style=Z dot] (65) at (9.25, -1) {};
		\node [style=X dot] (66) at (11.75, 0) {};
		\node [style=Z dot] (67) at (11.75, 1) {};
		\node [style=Z dot] (68) at (11.75, -1) {};
		\node [style=none] (70) at (18, -1) {};
		\node [style=none] (71) at (18, 1) {};
		\node [style=hadamard] (74) at (17.25, 0) {$i$};
		\node [style=none] (77) at (14.25, 0) {=};
		\node [style=Z phase dot] (78) at (15.5, 0) {$\frac\pi4$};
		\node [style=X dot] (80) at (16.25, 0) {};
		\node [style=Z dot] (81) at (17.25, 1) {};
		\node [style=Z dot] (82) at (17.25, -1) {};
		\node [style=none] (83) at (3.75, 1) {\eqref{eq:cat-CS-T}};
	\end{pgfonlayer}
	\begin{pgfonlayer}{edgelayer}
		\draw (2.center) to (1);
		\draw (0) to (3.center);
		\draw (4) to (6);
		\draw (6) to (5);
		\draw (8) to (7);
		\draw (45) to (44.center);
		\draw (46) to (43.center);
		\draw (50) to (49.center);
		\draw (51) to (48.center);
		\draw (52) to (50);
		\draw (54) to (57.center);
		\draw (58) to (60);
		\draw (60) to (59);
		\draw (62) to (61);
		\draw (65) to (56.center);
		\draw (67) to (66);
		\draw (66) to (68);
		\draw (66) to (64);
		\draw (81) to (80);
		\draw (80) to (82);
		\draw (80) to (78);
		\draw (74) to (82);
		\draw (81) to (74);
		\draw (81) to (71.center);
		\draw (70.center) to (82);
	\end{pgfonlayer}
\end{tikzpicture}
\end{equation}
We can then identify the underlying reason this catalysis works. It is because we have the following 'Euler decomposition' of the H-box with an $i$ phase:
\begin{equation}
    \begin{tikzpicture}
	\begin{pgfonlayer}{nodelayer}
		\node [style=none] (0) at (2.5, -1) {};
		\node [style=none] (1) at (2.5, 1) {};
		\node [style=hadamard] (2) at (1.75, 0) {$e^{i\pi/2}$};
		\node [style=none] (3) at (3.75, 0) {=};
		\node [style=Z dot] (6) at (1.75, 1) {};
		\node [style=Z dot] (7) at (1.75, -1) {};
		\node [style=none] (8) at (6.75, -1) {};
		\node [style=none] (9) at (6.75, 1) {};
		\node [style=Z phase dot] (11) at (6, 0) {$\ -\frac\pi4~$};
		\node [style=X dot] (12) at (5, 0) {};
		\node [style=Z phase dot] (13) at (5, 1) {$\frac\pi4$};
		\node [style=none] (15) at (-0.75, -1) {};
		\node [style=none] (16) at (-0.75, 1) {};
		\node [style=hadamard] (17) at (-1.5, 0) {$i$};
		\node [style=Z dot] (18) at (-1.5, 1) {};
		\node [style=Z dot] (19) at (-1.5, -1) {};
		\node [style=none] (20) at (0, 0) {=};
		\node [style=Z phase dot] (21) at (5, -1) {$\frac\pi4$};
		\node [style=none] (22) at (7.75, 0) {=};
		\node [style=none] (23) at (12.25, -1) {};
		\node [style=none] (24) at (12.25, 1) {};
		\node [style=hadamard] (25) at (12, 0) {$e^{-i\frac\pi4}$};
		\node [style=X dot] (26) at (10.5, 0) {};
		\node [style=Z dot] (27) at (10.5, 1) {};
		\node [style=Z dot] (28) at (10.5, -1) {};
		\node [style=hadamard] (29) at (9.25, 1) {$e^{i\frac\pi4}$};
		\node [style=hadamard] (30) at (9.25, -1) {$e^{i\frac\pi4}$};
	\end{pgfonlayer}
	\begin{pgfonlayer}{edgelayer}
		\draw (2) to (7);
		\draw (6) to (2);
		\draw (6) to (1.center);
		\draw (0.center) to (7);
		\draw (13) to (12);
		\draw (12) to (11);
		\draw (13) to (9.center);
		\draw (17) to (19);
		\draw (18) to (17);
		\draw (18) to (16.center);
		\draw (15.center) to (19);
		\draw (12) to (21);
		\draw (8.center) to (21);
		\draw (27) to (26);
		\draw (26) to (25);
		\draw (27) to (24.center);
		\draw (26) to (28);
		\draw (23.center) to (28);
		\draw (29) to (27);
		\draw (30) to (28);
	\end{pgfonlayer}
\end{tikzpicture}
\end{equation}
Here we use the fact that $\begin{tikzpicture}
	\begin{pgfonlayer}{nodelayer}
		\node [style=hadamard] (0) at (0, 0) {$e^{i\alpha}$};
		\node [style=none] (2) at (1.5, 0) {};
	\end{pgfonlayer}
	\begin{pgfonlayer}{edgelayer}
		\draw (0) to (2.center);
	\end{pgfonlayer}
\end{tikzpicture}
 \ = \ \begin{tikzpicture}
	\begin{pgfonlayer}{nodelayer}
		\node [style=white phase dot] (0) at (0, 0) {$\alpha$};
		\node [style=none] (2) at (1.25, 0) {};
	\end{pgfonlayer}
	\begin{pgfonlayer}{edgelayer}
		\draw (0) to (2.center);
	\end{pgfonlayer}
\end{tikzpicture}
$ to write the phases as H-boxes. We do this because such a rule doesn't just hold for an H-box with a label that is a complex phase like $e^{i\alpha}$, it in fact holds for any complex $a\neq 0$:
\begin{equation}\label{eq:H-box-a-Fourier}
    \begin{tikzpicture}
	\begin{pgfonlayer}{nodelayer}
		\node [style=none] (11) at (-0.75, -1) {};
		\node [style=none] (12) at (-0.75, 1) {};
		\node [style=hadamard] (13) at (-2, 0) {$a$};
		\node [style=none] (16) at (0, 0) {=};
		\node [style=none] (18) at (4.5, -1) {};
		\node [style=none] (19) at (4.5, 1) {};
		\node [style=hadamard] (20) at (4.25, 0) {$\frac{1}{\sqrt{a}}$};
		\node [style=X dot] (21) at (2.75, 0) {};
		\node [style=Z dot] (22) at (2.75, 1) {};
		\node [style=Z dot] (23) at (2.75, -1) {};
		\node [style=hadamard] (24) at (1.5, 1) {$\sqrt{a}$};
		\node [style=hadamard] (25) at (1.5, -1) {$\sqrt{a}$};
	\end{pgfonlayer}
	\begin{pgfonlayer}{edgelayer}
		\draw (22) to (21);
		\draw (21) to (20);
		\draw (22) to (19.center);
		\draw (21) to (23);
		\draw (18.center) to (23);
		\draw (24) to (22);
		\draw (25) to (23);
		\draw [in=180, out=-90, looseness=1.25] (13) to (11.center);
		\draw [in=180, out=90] (13) to (12.center);
	\end{pgfonlayer}
\end{tikzpicture}
\end{equation}
This then allows us to write down a generalisation of Eq.~\eqref{eq:cat-CS-T-ZH} to arbitrary H-boxes:
\begin{equation}\label{eq:H-box-catalysis}
    \begin{tikzpicture}
	\begin{pgfonlayer}{nodelayer}
		\node [style=none] (0) at (-1.25, -1) {};
		\node [style=none] (1) at (-1.25, 1) {};
		\node [style=hadamard] (2) at (-2, 0) {$a^2$};
		\node [style=none] (3) at (0, 0) {=};
		\node [style=hadamard] (4) at (-4, 0) {$a$};
		\node [style=X dot] (5) at (-3, 0) {};
		\node [style=Z dot] (6) at (-2, 1) {};
		\node [style=Z dot] (7) at (-2, -1) {};
		\node [style=hadamard] (11) at (1.25, 0) {$a$};
		\node [style=X dot] (12) at (2.25, 0) {};
		\node [style=none] (15) at (5, -1) {};
		\node [style=none] (16) at (5, 1) {};
		\node [style=hadamard] (17) at (4.25, 0) {$\frac1a$};
		\node [style=X dot] (18) at (3.25, 0) {};
		\node [style=Z dot] (19) at (3.25, 1) {};
		\node [style=Z dot] (20) at (3.25, -1) {};
		\node [style=hadamard] (21) at (3.25, 1.75) {$a$};
		\node [style=hadamard] (22) at (3.25, -1.75) {$a$};
		\node [style=none] (23) at (0, 1) {\eqref{eq:H-box-a-Fourier}};
		\node [style=none] (24) at (6, 0) {=};
		\node [style=hadamard] (25) at (7.5, 1) {$a$};
		\node [style=none] (27) at (11, -1) {};
		\node [style=none] (28) at (11, 1) {};
		\node [style=hadamard] (29) at (7.5, -1) {$\frac1a$};
		\node [style=X dot] (30) at (9.25, 0) {};
		\node [style=Z dot] (31) at (9.25, 1) {};
		\node [style=Z dot] (32) at (9.25, -1) {};
		\node [style=hadamard] (33) at (10, 1.75) {$a$};
		\node [style=hadamard] (34) at (10, -1.75) {$a$};
		\node [style=Z dot] (35) at (8.25, 0) {};
		\node [style=Z dot] (36) at (10, 1) {};
		\node [style=Z dot] (37) at (10, -1) {};
		\node [style=none] (38) at (6, 2) {\spiderrule};
		\node [style=none] (39) at (6, 1) {\bialgrule};
		\node [style=none] (40) at (12.5, 0) {=};
		\node [style=none] (42) at (16.75, -1) {};
		\node [style=none] (43) at (16.75, 1) {};
		\node [style=hadamard] (44) at (14, 0) {$1$};
		\node [style=X dot] (45) at (15, 0) {};
		\node [style=Z dot] (46) at (15, 1) {};
		\node [style=Z dot] (47) at (15, -1) {};
		\node [style=hadamard] (48) at (15.75, 1.75) {$a$};
		\node [style=hadamard] (49) at (15.75, -1.75) {$a$};
		\node [style=Z dot] (51) at (15.75, 1) {};
		\node [style=Z dot] (52) at (15.75, -1) {};
		\node [style=none] (53) at (12.5, 1) {\eqref{eq:multiply-rule-phased}};
		\node [style=none] (54) at (18, 1) {\eqref{eq:Z-a-state}};
		\node [style=none] (55) at (18, 0) {=};
		\node [style=none] (56) at (22.25, -0.75) {};
		\node [style=none] (57) at (22.25, 0.75) {};
		\node [style=Z dot] (58) at (19.5, 0) {};
		\node [style=X dot] (59) at (20.5, 0) {};
		\node [style=Z dot] (60) at (20.5, 0.75) {};
		\node [style=Z dot] (61) at (20.5, -0.75) {};
		\node [style=hadamard] (62) at (21.25, 1.5) {$a$};
		\node [style=hadamard] (63) at (21.25, -1.5) {$a$};
		\node [style=Z dot] (64) at (21.25, 0.75) {};
		\node [style=Z dot] (65) at (21.25, -0.75) {};
		\node [style=none] (67) at (24, 0) {=};
		\node [style=none] (68) at (26.5, -0.75) {};
		\node [style=none] (69) at (26.5, 0.75) {};
		\node [style=hadamard] (74) at (25.5, 0.75) {$a$};
		\node [style=hadamard] (75) at (25.5, -0.75) {$a$};
	\end{pgfonlayer}
	\begin{pgfonlayer}{edgelayer}
		\draw (6) to (5);
		\draw (5) to (7);
		\draw (5) to (4);
		\draw (2) to (7);
		\draw (6) to (2);
		\draw (6) to (1.center);
		\draw (0.center) to (7);
		\draw (12) to (11);
		\draw (19) to (18);
		\draw (18) to (17);
		\draw (19) to (16.center);
		\draw (18) to (20);
		\draw (15.center) to (20);
		\draw (21) to (19);
		\draw (22) to (20);
		\draw (19) to (12);
		\draw (12) to (20);
		\draw (31) to (30);
		\draw (31) to (28.center);
		\draw (30) to (32);
		\draw (27.center) to (32);
		\draw (35) to (30);
		\draw (25) to (35);
		\draw (29) to (35);
		\draw (33) to (36);
		\draw (37) to (34);
		\draw (46) to (45);
		\draw (46) to (43.center);
		\draw (45) to (47);
		\draw (42.center) to (47);
		\draw (48) to (51);
		\draw (52) to (49);
		\draw (44) to (45);
		\draw (60) to (59);
		\draw (60) to (57.center);
		\draw (59) to (61);
		\draw (56.center) to (61);
		\draw (62) to (64);
		\draw (65) to (63);
		\draw (58) to (59);
		\draw (75) to (68.center);
		\draw (74) to (69.center);
	\end{pgfonlayer}
\end{tikzpicture}
\end{equation}
Note that this equation first appeared in~\cite{EPTCS384.6}.
Here we wrote $a^2$ in the 2-ary H-box instead of $a$ so that we don't have to work with square roots. When we take $a=e^{i\frac\pi4}$ we get Eq.~\eqref{eq:cat-CS-T-ZH}, but this works for any value. A particularly simple, but still interesting case is when $a=e^{i\frac\pi2} = i$. Translating this back into circuit form gives us a catalysis of $\ket{i} := \ket{0} + i\ket{1}$ states using a CZ. As a rewrite rule this is essentially equivalent to the Euler decomposition of a Hadamard. While this rule itself is simple, it already demonstrates the power of the catalysis framework in proving completeness.

The phase-free ZH-calculus is complete for the ring $\Z[\frac12]$.
By adding the generator $\begin{tikzpicture}
	\begin{pgfonlayer}{nodelayer}
		\node [style=none] (69) at (1, 0) {};
		\node [style=hadamard] (74) at (0, 0) {$i$};
	\end{pgfonlayer}
	\begin{pgfonlayer}{edgelayer}
		\draw (74) to (69.center);
	\end{pgfonlayer}
\end{tikzpicture}
$, a single-ary H-box with label $a=i$, we get a universal representation for the ring $\Z[\frac12,i]$~\cite{zhcompleteness2020}. As it turns out, adding the rule Eq.~\eqref{eq:H-box-catalysis} for $a=i$ to the already existing rules for the phase-free fragment is already \emph{almost} enough to get a complete calculus for this bigger fragment $\Z[\frac12,i]$ which includes $\begin{tikzpicture}
	\begin{pgfonlayer}{nodelayer}
		\node [style=none] (69) at (1, 0) {};
		\node [style=hadamard] (74) at (0, 0) {$i$};
	\end{pgfonlayer}
	\begin{pgfonlayer}{edgelayer}
		\draw (74) to (69.center);
	\end{pgfonlayer}
\end{tikzpicture}
$.

To see this, we first consider what a generic diagram in the $\Z[\frac12,i]$ fragment looks like. We added the generator $\begin{tikzpicture}
	\begin{pgfonlayer}{nodelayer}
		\node [style=none] (69) at (1, 0) {};
		\node [style=hadamard] (74) at (0, 0) {$i$};
	\end{pgfonlayer}
	\begin{pgfonlayer}{edgelayer}
		\draw (74) to (69.center);
	\end{pgfonlayer}
\end{tikzpicture}
$, so now a diagram consists of generators from the old $\Z[\frac12]$ fragment plus this new generator, used an arbitrary number of times. Using Eq.~\eqref{eq:H-box-catalysis} we can however reduce all these separate instances of $\begin{tikzpicture}
	\begin{pgfonlayer}{nodelayer}
		\node [style=none] (69) at (1, 0) {};
		\node [style=hadamard] (74) at (0, 0) {$i$};
	\end{pgfonlayer}
	\begin{pgfonlayer}{edgelayer}
		\draw (74) to (69.center);
	\end{pgfonlayer}
\end{tikzpicture}
$ into just one of them, reducing the complexity of the diagram. That is, given some diagram $D$ in the $\Z[i]$ fragment, we can rewrite it to a diagram $D'$ containing just generators from the $\Z[\frac12]$ fragment such that:
\begin{equation}\label{eq:diagram-catalysis-i}
    \begin{tikzpicture}
	\begin{pgfonlayer}{nodelayer}
		\node [style=medium box] (0) at (-2.5, 0) {$D$};
		\node [style=none] (1) at (-4, 0.75) {};
		\node [style=none] (2) at (-1, 0.75) {};
		\node [style=none] (3) at (-4, -0.75) {};
		\node [style=none] (4) at (-1, -0.75) {};
		\node [style=none] (5) at (-3.5, 0.25) {$\vdots$};
		\node [style=none] (6) at (-1.5, 0.25) {$\vdots$};
		\node [style=none] (7) at (0, 0) {=};
		\node [style=medium box] (8) at (2.5, 0) {$D'$};
		\node [style=none] (9) at (1, 0.75) {};
		\node [style=none] (10) at (4, 0.75) {};
		\node [style=none] (11) at (1, -0.75) {};
		\node [style=none] (12) at (4, -0.75) {};
		\node [style=none] (13) at (1.5, 0.25) {$\vdots$};
		\node [style=none] (14) at (3.5, 0.25) {$\vdots$};
		\node [style=hadamard] (15) at (2.5, -1.5) {$i$};
		\node [style=none] (16) at (0, 1) {\eqref{eq:H-box-catalysis}};
	\end{pgfonlayer}
	\begin{pgfonlayer}{edgelayer}
		\draw (3.center) to (4.center);
		\draw (2.center) to (1.center);
		\draw (11.center) to (12.center);
		\draw (10.center) to (9.center);
		\draw (8) to (15);
	\end{pgfonlayer}
\end{tikzpicture}
\end{equation}
We are for now ignoring the edge case where the original diagram did not contain any instance of $\begin{tikzpicture}
	\begin{pgfonlayer}{nodelayer}
		\node [style=none] (69) at (1, 0) {};
		\node [style=hadamard] (74) at (0, 0) {$i$};
	\end{pgfonlayer}
	\begin{pgfonlayer}{edgelayer}
		\draw (74) to (69.center);
	\end{pgfonlayer}
\end{tikzpicture}
$.

As a shorthand, we will write $D'[\ket{\psi}]$ for the diagram we get when we plug $\ket{\psi}$ into the bottom input in Eq.~\eqref{eq:diagram-catalysis-i}. So here we have $D = D'[\begin{tikzpicture}
	\begin{pgfonlayer}{nodelayer}
		\node [style=none] (69) at (1, 0) {};
		\node [style=hadamard] (74) at (0, 0) {$i$};
	\end{pgfonlayer}
	\begin{pgfonlayer}{edgelayer}
		\draw (74) to (69.center);
	\end{pgfonlayer}
\end{tikzpicture}
]$.
Note that $\begin{tikzpicture}
	\begin{pgfonlayer}{nodelayer}
		\node [style=none] (69) at (1, 0) {};
		\node [style=hadamard] (74) at (0, 0) {$i$};
	\end{pgfonlayer}
	\begin{pgfonlayer}{edgelayer}
		\draw (74) to (69.center);
	\end{pgfonlayer}
\end{tikzpicture}
 = \ket{0} + i\ket{1}$. Hence, if we expand it like this we see that $D$ is equal to a sum of two diagrams: $D'$ where we plugged in $\ket{0}$ into the bottom wire, and $iD'$ where we plugged $\ket{1}$ into the bottom wire: $D = D'[\ket{0}] + iD'[\ket{1}]$.

Now suppose we have two diagrams $D_1$ and $D_2$ in the $\Z[i]$ fragment and that they implement the same linear map: $D_1 = D_2$. We can both decompose them as described above to get $D'_1[\ket{0}] + iD'_1[\ket{1}] = D'_2[\ket{0}] + i D'_2[\ket{1}]$. Each of these $D'_j[\ket{x}]$ diagrams represents a matrix that is entirely real-valued, so the only way for this equation of complex matrices to hold, is if it holds for the real part and for the complex part separately:
\begin{equation}
    D'_1[\ket{0}] \ =\  D'_2[\ket{0}] \qquad \quad D'_1[\ket{1}] \ =\  D'_2[\ket{1}]
\end{equation}
We then conclude that $D'_1$ and $D'_2$ are equal when we input either $\ket{0}$ or $\ket{1}$ into the bottom wire. As these states form a basis, this must then hold for any input. We can then leave this wire open and still have an equality:
\begin{equation}\label{eq:diagram-catalysis-equality}
    \begin{tikzpicture}
	\begin{pgfonlayer}{nodelayer}
		\node [style=none] (7) at (0, 0) {=};
		\node [style=medium box] (8) at (-2.5, 0) {$D'_1$};
		\node [style=none] (9) at (-4, 0.75) {};
		\node [style=none] (10) at (-1, 0.75) {};
		\node [style=none] (11) at (-4, -0.75) {};
		\node [style=none] (12) at (-1, -0.75) {};
		\node [style=none] (13) at (-3.5, 0.25) {$\vdots$};
		\node [style=none] (14) at (-1.5, 0.25) {$\vdots$};
		\node [style=none] (15) at (-2.5, -1.5) {};
		\node [style=medium box] (16) at (2.5, 0) {$D'_2$};
		\node [style=none] (17) at (1, 0.75) {};
		\node [style=none] (18) at (4, 0.75) {};
		\node [style=none] (19) at (1, -0.75) {};
		\node [style=none] (20) at (4, -0.75) {};
		\node [style=none] (21) at (1.5, 0.25) {$\vdots$};
		\node [style=none] (22) at (3.5, 0.25) {$\vdots$};
		\node [style=none] (23) at (2.5, -1.5) {};
	\end{pgfonlayer}
	\begin{pgfonlayer}{edgelayer}
		\draw (11.center) to (12.center);
		\draw (10.center) to (9.center);
		\draw (8) to (15.center);
		\draw (19.center) to (20.center);
		\draw (18.center) to (17.center);
		\draw (16) to (23.center);
	\end{pgfonlayer}
\end{tikzpicture}
\end{equation}
We have this equality as linear maps, but both diagrams are in the $\Z[\frac12]$ fragment for which we have completeness. We hence know how to rewrite one into the other using the rules of the phase-free ZH-calculus. This gives us then a path to rewrite the original $D_1$ into $D_2$:
\begin{equation}\label{eq:diagram-catalysis-completeness}
    \begin{tikzpicture}
	\begin{pgfonlayer}{nodelayer}
		\node [style=none] (0) at (0, 0) {=};
		\node [style=medium box] (1) at (-2.5, 0) {$D'_1$};
		\node [style=none] (2) at (-4, 0.75) {};
		\node [style=none] (3) at (-1, 0.75) {};
		\node [style=none] (4) at (-4, -0.75) {};
		\node [style=none] (5) at (-1, -0.75) {};
		\node [style=none] (6) at (-3.5, 0.25) {$\vdots$};
		\node [style=none] (7) at (-1.5, 0.25) {$\vdots$};
		\node [style=hadamard] (8) at (-2.5, -1.5) {$i$};
		\node [style=medium box] (9) at (2.5, 0) {$D'_2$};
		\node [style=none] (10) at (1, 0.75) {};
		\node [style=none] (11) at (4, 0.75) {};
		\node [style=none] (12) at (1, -0.75) {};
		\node [style=none] (13) at (4, -0.75) {};
		\node [style=none] (14) at (1.5, 0.25) {$\vdots$};
		\node [style=none] (15) at (3.5, 0.25) {$\vdots$};
		\node [style=hadamard] (16) at (2.5, -1.5) {$i$};
		\node [style=none] (17) at (0, 1) {(*)};
		\node [style=medium box] (18) at (-7.5, 0) {$D_1$};
		\node [style=none] (19) at (-9, 0.75) {};
		\node [style=none] (20) at (-6, 0.75) {};
		\node [style=none] (21) at (-9, -0.75) {};
		\node [style=none] (22) at (-6, -0.75) {};
		\node [style=none] (23) at (-8.5, 0.25) {$\vdots$};
		\node [style=none] (24) at (-6.5, 0.25) {$\vdots$};
		\node [style=none] (25) at (-5, 0) {=};
		\node [style=medium box] (26) at (7, 0) {$D_2$};
		\node [style=none] (27) at (5.5, 0.75) {};
		\node [style=none] (28) at (8.5, 0.75) {};
		\node [style=none] (29) at (5.5, -0.75) {};
		\node [style=none] (30) at (8.5, -0.75) {};
		\node [style=none] (31) at (6, 0.25) {$\vdots$};
		\node [style=none] (32) at (8, 0.25) {$\vdots$};
		\node [style=none] (33) at (4.75, 0) {=};
		\node [style=none] (34) at (-5, 1) {\eqref{eq:H-box-catalysis}};
		\node [style=none] (35) at (4.75, 1) {\eqref{eq:H-box-catalysis}};
	\end{pgfonlayer}
	\begin{pgfonlayer}{edgelayer}
		\draw (4.center) to (5.center);
		\draw (3.center) to (2.center);
		\draw (1) to (8);
		\draw (12.center) to (13.center);
		\draw (11.center) to (10.center);
		\draw (9) to (16);
		\draw (21.center) to (22.center);
		\draw (20.center) to (19.center);
		\draw (29.center) to (30.center);
		\draw (28.center) to (27.center);
	\end{pgfonlayer}
\end{tikzpicture}
\end{equation}
Here each equality is now a diagrammatic equality, and with (*) we denote we are using rewrites from the original complete calculus for the $\Z$ fragment. This would give us completeness for the fragment $\Z[\frac12,i]$, except that we have ignored an edge case.
We can only rewrite a diagram in the $\Z[i]$ fragment as in Eq.~\eqref{eq:diagram-catalysis-i} if there is at least one generator \begin{tikzpicture}
	\begin{pgfonlayer}{nodelayer}
		\node [style=none] (69) at (1, 0) {};
		\node [style=hadamard] (74) at (0, 0) {$i$};
	\end{pgfonlayer}
	\begin{pgfonlayer}{edgelayer}
		\draw (74) to (69.center);
	\end{pgfonlayer}
\end{tikzpicture}
 present in the diagram. In fact, we currently haven't assumed any rewrite rule that relate a diagram containing a \begin{tikzpicture}
	\begin{pgfonlayer}{nodelayer}
		\node [style=none] (69) at (1, 0) {};
		\node [style=hadamard] (74) at (0, 0) {$i$};
	\end{pgfonlayer}
	\begin{pgfonlayer}{edgelayer}
		\draw (74) to (69.center);
	\end{pgfonlayer}
\end{tikzpicture}
 to one that does not contain any \begin{tikzpicture}
	\begin{pgfonlayer}{nodelayer}
		\node [style=none] (69) at (1, 0) {};
		\node [style=hadamard] (74) at (0, 0) {$i$};
	\end{pgfonlayer}
	\begin{pgfonlayer}{edgelayer}
		\draw (74) to (69.center);
	\end{pgfonlayer}
\end{tikzpicture}
. This means in particular that our currently considered rule set cannot prove the following true equation:
\begin{equation}\label{eq:X-H-i-empty}
    \begin{tikzpicture}
	\begin{pgfonlayer}{nodelayer}
		\node [style=X dot] (0) at (-2, 0) {};
		\node [style=hadamard] (1) at (-1, 0) {$i$};
		\node [style=none] (2) at (0, 0) {=};
		\node [style=X dot] (3) at (1, 0) {};
		\node [style=Z dot] (4) at (2, 0) {};
	\end{pgfonlayer}
	\begin{pgfonlayer}{edgelayer}
		\draw (0) to (1);
		\draw (3) to (4);
	\end{pgfonlayer}
\end{tikzpicture}

\end{equation}
However, when we also add Eq.~\eqref{eq:X-H-i-empty} as an additional rule, then this problem is solved and it \emph{is} true that we can then always rewrite a diagram in the $\Z[i]$ fragment as in Eq.~\eqref{eq:diagram-catalysis-i}: if the diagram contains at least one \begin{tikzpicture}
	\begin{pgfonlayer}{nodelayer}
		\node [style=none] (69) at (1, 0) {};
		\node [style=hadamard] (74) at (0, 0) {$i$};
	\end{pgfonlayer}
	\begin{pgfonlayer}{edgelayer}
		\draw (74) to (69.center);
	\end{pgfonlayer}
\end{tikzpicture}
 we can already use Eq.~\eqref{eq:H-box-catalysis} to transform to the form of Eq.~\eqref{eq:diagram-catalysis-i}, and if it does not, we can use Eq.~\eqref{eq:X-H-i-empty} once to introduce one \begin{tikzpicture}
	\begin{pgfonlayer}{nodelayer}
		\node [style=none] (69) at (1, 0) {};
		\node [style=hadamard] (74) at (0, 0) {$i$};
	\end{pgfonlayer}
	\begin{pgfonlayer}{edgelayer}
		\draw (74) to (69.center);
	\end{pgfonlayer}
\end{tikzpicture}
, in which case it is also in the form of Eq.~\eqref{eq:diagram-catalysis-i}.

\begin{proposition}
    The graphical calculus consisting of the phase-free ZH generators and $\begin{tikzpicture}
	\begin{pgfonlayer}{nodelayer}
		\node [style=none] (69) at (1, 0) {};
		\node [style=hadamard] (74) at (0, 0) {$i$};
	\end{pgfonlayer}
	\begin{pgfonlayer}{edgelayer}
		\draw (74) to (69.center);
	\end{pgfonlayer}
\end{tikzpicture}
$, together with the phase-free rewrite rules of Figure~\ref{fig:phasefree-rules} augmented with the catalysis rule Eq.~\eqref{eq:H-box-catalysis} for $a=i$, and the rule \begin{tikzpicture}
	\begin{pgfonlayer}{nodelayer}
		\node [style=X dot] (0) at (-2, 0) {};
		\node [style=hadamard] (1) at (-1, 0) {$i$};
		\node [style=none] (2) at (0, 0) {=};
		\node [style=X dot] (3) at (1, 0) {};
		\node [style=Z dot] (4) at (2, 0) {};
	\end{pgfonlayer}
	\begin{pgfonlayer}{edgelayer}
		\draw (0) to (1);
		\draw (3) to (4);
	\end{pgfonlayer}
\end{tikzpicture}
 is complete and universal for the ring $\Z[\frac12,i]$.
\end{proposition} 

This trick for extending the calculus doesn't just work for $i$: it works for any complex number $a\neq 0$ such that $a^2\in \Z$ using a very similar argument. We can also iterate it: once we have a calculus complete and universal for $\Z[\frac12,a]$ we can pick any $b\neq 0$ such that $b^2\in \Z[\frac12,a]$ and augment the calculus with the appropriate catalysis and scalar-introduction rule to get a new complete calculus for $\Z[\frac12,a,b]$.
We present the proof in Appendix~\ref{app:completeness}.
\begin{theorem}\label{thm:completeness-extension}
    Let $a_1,\ldots, a_k$ be a series of non-zero complex numbers such that $a_j^2\in \Z[\frac12, a_1,\ldots, a_{j-1}]$. Then the phase-free ZH-calculus augmented with generators \begin{tikzpicture}
	\begin{pgfonlayer}{nodelayer}
		\node [style=none] (0) at (1.5, 0) {};
		\node [style=hadamard] (1) at (0.5, 0) {$a_j$};
	\end{pgfonlayer}
	\begin{pgfonlayer}{edgelayer}
		\draw (1) to (0.center);
	\end{pgfonlayer}
\end{tikzpicture}
 and the following rules is complete for the ring $\Z[\frac12, a_1,\ldots, a_k]$:
    \begin{equation*}
        \begin{tikzpicture}
	\begin{pgfonlayer}{nodelayer}
		\node [style=none] (0) at (-1, -1) {};
		\node [style=none] (1) at (-1, 1) {};
		\node [style=hadamard] (2) at (-1.75, 0) {$a^2$};
		\node [style=none] (3) at (0, 0) {=};
		\node [style=hadamard] (4) at (-3.75, 0) {$a$};
		\node [style=X dot] (5) at (-2.75, 0) {};
		\node [style=Z dot] (6) at (-1.75, 1) {};
		\node [style=Z dot] (7) at (-1.75, -1) {};
		\node [style=none] (59) at (2.15, -0.975) {};
		\node [style=none] (60) at (2.15, 0.975) {};
		\node [style=hadamard] (61) at (1.15, 0.975) {$a$};
		\node [style=hadamard] (62) at (1.15, -0.975) {$a$};
	\end{pgfonlayer}
	\begin{pgfonlayer}{edgelayer}
		\draw (6) to (5);
		\draw (5) to (7);
		\draw (5) to (4);
		\draw (2) to (7);
		\draw (6) to (2);
		\draw (6) to (1.center);
		\draw (0.center) to (7);
		\draw (62) to (59.center);
		\draw (61) to (60.center);
	\end{pgfonlayer}
\end{tikzpicture} \quad \text{and} \quad \begin{tikzpicture}
	\begin{pgfonlayer}{nodelayer}
		\node [style=X dot] (0) at (-2.375, 0) {};
		\node [style=hadamard] (1) at (-1.25, 0) {$a$};
		\node [style=none] (2) at (0, 0) {=};
		\node [style=X dot] (3) at (1, 0) {};
		\node [style=Z dot] (4) at (2, 0) {};
	\end{pgfonlayer}
	\begin{pgfonlayer}{edgelayer}
		\draw (0) to (1);
		\draw (3) to (4);
	\end{pgfonlayer}
\end{tikzpicture}
 \quad \text{ for all } a=a_j
    \end{equation*}
    Here the $a^2$ H-box should be understood as short-hand for some diagram in the smaller fragment representing the matrix for that H-box.
\end{theorem}
    
When we take $a_1 = i$ and $a_2 = e^{i\frac\pi4}$ we get the ring $\Z[\frac12,i,e^{i\frac\pi4}] = \Z[i,\frac{1}{\sqrt{2}}]$ corresponding to Clifford+T computation, and in this case we can simplify the rules a bit more to get a simple axiomatisation of the Clifford+T maps.
\begin{proposition}
    The phase-free ZH-calculus augmented with H-boxes with a label of $i$ and $e^{i\frac\pi4}$ and the following rules 
    is complete for matrices with entries in the ring $\Z[i,\frac{1}{\sqrt{2}}]$:
    \begin{equation*}
        \begin{tikzpicture}
	\begin{pgfonlayer}{nodelayer}
		\node [style=none] (0) at (-7.5, 0) {};
		\node [style=none] (1) at (-4.5, 0) {};
		\node [style=hadamard] (2) at (-6, 0) {};
		\node [style=none] (3) at (-3.5, 0) {=};
		\node [style=none] (4) at (-2.5, 0) {};
		\node [style=none] (5) at (1.25, 0) {};
		\node [style=Z phase dot] (6) at (-1.25, 0) {$\frac\pi2$};
		\node [style=Z phase dot] (7) at (0.25, 0) {$\frac\pi2$};
		\node [style=Z dot] (10) at (-6.75, 0) {};
		\node [style=Z dot] (11) at (-5.25, 0) {};
		\node [style=X dot] (12) at (-6, 1) {};
		\node [style=Z phase dot] (13) at (-6, 1.75) {$\frac\pi2$};
		\node [style=none] (14) at (4, 0) {};
		\node [style=none] (15) at (7, 0) {};
		\node [style=hadamard] (16) at (5.5, 0) {$i$};
		\node [style=none] (17) at (7.5, 0) {=};
		\node [style=none] (18) at (8.5, 0) {};
		\node [style=none] (19) at (12.25, 0) {};
		\node [style=Z phase dot] (20) at (9.75, 0) {$\frac\pi4$};
		\node [style=Z phase dot] (21) at (11.25, 0) {$\frac\pi4$};
		\node [style=Z dot] (22) at (4.75, 0) {};
		\node [style=Z dot] (23) at (6.25, 0) {};
		\node [style=X dot] (24) at (5.5, 1) {};
		\node [style=Z phase dot] (25) at (5.5, 1.75) {$\frac\pi4$};
		\node [style=Z phase dot] (26) at (15, 0) {$\frac\pi4$};
		\node [style=X dot] (27) at (16.25, 0) {};
		\node [style=none] (28) at (17.5, 0) {};
		\node [style=none] (29) at (17.5, 0) {=};
		\node [style=Z dot] (30) at (18.75, 0) {};
		\node [style=X dot] (31) at (20, 0) {};
	\end{pgfonlayer}
	\begin{pgfonlayer}{edgelayer}
		\draw (0.center) to (1.center);
		\draw (12) to (10);
		\draw (12) to (11);
		\draw (13) to (12);
		\draw (4.center) to (6);
		\draw (7) to (5.center);
		\draw (14.center) to (15.center);
		\draw (24) to (22);
		\draw (24) to (23);
		\draw (25) to (24);
		\draw (18.center) to (20);
		\draw (21) to (19.center);
		\draw (26) to (27);
		\draw (30) to (31);
	\end{pgfonlayer}
\end{tikzpicture}
    \end{equation*}
\end{proposition}
\begin{proof}
    These are the additional rules needed by Theorem~\ref{thm:completeness-extension} for completeness, except we don't have the scalar introduction rule for the label $i$. This rule can be derived from the $e^{i\frac\pi4}$ one, in combination with the catalysis rule for $e^{i\frac\pi4}$.
\end{proof}
Here we presented the rules in a slightly different manner to make clear the connection between catalysis for $\frac\pi2$ and the standard Euler decomposition of the Hadamard. Continuing the division into smaller dyadic rational multiples of $\pi$ of the form $e^{i\frac{\pi}{2^k}}$, we see that we can get a complete calculus for diagrams corresponding to Clifford-cyclotomic circuits, similar to how exact synthesis for these circuits was proven in~\cite{amy2023b}. Independently to our results, completeness for dyadic angles was also shown in the context of the Sum-over-paths formalism in~\cite{Vilmart2023Completeness} using a technique that is reminiscent of catalysis.

While there have been previous complete graphical calculi for the fragment corresponding to Clifford+$T$ circuits~\cite{SimonCompleteness,vilmartzxtriangle}, the result we find here has the double benefit of having easy to interpret axioms, and a generic proof. The axioms consist of the phase-free ZH ones, each of which corresponds to a simple property of the Boolean maps COPY, XOR and AND~\cite{zhcompleteness2020}, plus the catalysis rules. This extension is not specific to Clifford+$T$ and works for any ring extension of the form stated in Theorem~\ref{thm:completeness-extension}. Note that~\cite{zhcompleteness2020} also describes how the ZH-calculus can be made complete for arbitrary rings, but these require adding three families of rules that are each parametrised over all the elements in the ring, and hence gives a much more complex rule set.


\section{Conclusion}\label{sec:conclusion}

We applied the technique of quantum state catalysis to prove generic results in universality and completeness. In particular, we obtained new proofs of the universality of the CS+Hadamard and Toffoli+Hadamard gate sets, as well as a new simple graphical calculus for the Clifford+$T$ fragment. Our results simplify the original proofs of these statements considerably, and pave the way for further applications of catalysis in these areas.
Notably, our completeness result did not use any special property of the ZH-calculus. The fact that it was universal for a large enough fragment to support catalysis was enough to find completeness of an extension. This technique could hence also be applied to study calculi over qudits, where for certain interesting fragments of quantum computing (like qudit Clifford+$T$), there is still no complete calculus. In a related direction, one could consider using catalytic methods to study mixed-dimensional calculi, since catalytic constructions such as the ones in \cite{amy2023a} can often be made to take advantage of mixed-dimensional ancillae.


\bibliographystyle{eptcs}
\bibliography{main}

\appendix

\section{Proofs}

\subsection{Universality of real-valued unitaries}\label{app:real-valued}

\begin{lemma}
    $\tilde{U}$ is indeed unitary for any choice of $U$.
\end{lemma}
\begin{proof}
    We will first show the following claims:
    \begin{enumerate}[a)]
        \item We can express $\Re(U^\dagger)$ and $\Im(U^\dagger)$ in terms of $\Re(U)$ and $\Im(U)$.
        \item We can Express $\Re(UV)$ and $\Im(UV)$ in terms of $\Re(U)$, $\Re(V)$, $\Im(U)$ and $\Im(V)$.
        \item We have $((\bra{0}\otimes \bra{\psi})\tilde{U}^\dagger) (\tilde{U}(\ket{0}\otimes \ket{\psi'})) = \braket{\psi}{\psi'}$. 
        \item We have $((\bra{0}\otimes \bra{\psi})\tilde{U}^\dagger) (\tilde{U}(\ket{1}\otimes \ket{\psi'})) = 0$. 
    \end{enumerate}
    \begin{enumerate}[a)]
        \item $\Re(U^\dagger) = \Re(U)^\dagger$. $\Im(U^\dagger) = - \Im(U)^\dagger$.
        \item $\Re(UV) = \Re(U)\Re(V) - \Im(U)\Im(V)$. $\Im(UV) = \Re(U)\Im(V) + \Im(U)\Re(V)$.
        \item First note that $(\bra{0}\otimes \bra{\psi})\tilde{U}^\dagger = (\tilde{U}(\ket{0}\otimes \ket{\psi}))^\dagger = \bra{0}\otimes (\bra(\psi) \Re(U^\dagger)) - \bra{1} \otimes (\bra{\psi}\Im(U^\dagger))$. Hence, using $\braket{0}{1} = 0$ the inner product reduces to $\bra{\psi} \Re(U^\dagger) \Re(U) \ket{\psi'} - \bra{\psi} \Im(U^\dagger) \Im(U) \ket{\psi'} = \bra{\psi}(\Re(U^\dagger) \Re(U) - \Im(U^\dagger) \Im(U)) \ket{\psi'} = \bra{\psi} \Re(U^\dagger U) \ket{\psi'} = \bra{\psi} \Re(I) \ket{\psi'} = \braket{\psi}{\psi'}$.
        \item Similar to the above.
    \end{enumerate}
    Let $\ket{\psi_k}$ form an orthogonal basis of $n$-qubit state space. The last two points above show that $\tilde{U}$ preserves the orthogonality of $\{\ket{\psi_k}\otimes \ket{0}, \ket{\psi_k}\otimes \ket{1}\}$. Hence, since it sends an orthogonal basis to an orthogonal basis, it is unitary.
\end{proof}

The encoding into $\tilde{U}$ is compositional, meaning we can apply it iteratively to a sequence of unitaries.
\begin{lemma}
    $\widetilde{UV} = \tilde{U}\tilde{V}$.
\end{lemma}
\begin{proof}
    Proven easily by making a case distinction on input states $\ket{\psi}\otimes \ket{0}$ and $\ket{\psi}\otimes \ket{1}$.
\end{proof}

\subsection{Generic universality through catalysis}\label{app:proof-catalysis}

\begin{proof}[Proof of Theorem~\ref{thm:catalysis-universality}]
    Let $C$ be some circuit over $U$. Then $C':= \phi(C)$ is a circuit over $\mathcal{V}$ such that $C'\ket{\psi}\otimes \ket{c} = (C\ket{\psi})\otimes \ket{c}$. By assumption we have $\ket{\psi_j} = C_j\ket{0\cdots 0}$ for some circuit $C_j$ over $\mathcal{V}$. Hence, for some observable $\mathcal{O}$ we have:
    $$\begin{tikzpicture}
	\begin{pgfonlayer}{nodelayer}
		\node [style=none] (31) at (-1, -1) {=};
		\node [style=wide left state] (32) at (-8.75, -1) {$\psi$};
		\node [style=big gate] (33) at (-7.25, -1) {$C$};
		\node [style=big gate] (34) at (-5.75, -1) {$\mathcal{O}$};
		\node [style=big gate] (35) at (-4.25, -1) {$C^\dagger$};
		\node [style=wide right state] (36) at (-2.375, -1) {$\psi$};
		\node [style=none] (37) at (-8.5, -0.25) {};
		\node [style=none] (38) at (-2.625, -0.25) {};
		\node [style=none] (39) at (-8.5, -1.75) {};
		\node [style=none] (40) at (-2.625, -1.75) {};
		\node [style=none] (41) at (-8, -0.75) {$\vdots$};
		\node [style=none] (42) at (-3.125, -0.75) {$\vdots$};
		\node [style=wide left state] (52) at (0.5, 0) {$\psi$};
		\node [style=big gate] (53) at (2, 0) {$C$};
		\node [style=big gate] (54) at (3.5, 0) {$\mathcal{O}$};
		\node [style=big gate] (55) at (5, 0) {$C^\dagger$};
		\node [style=wide right state] (56) at (6.875, 0) {$\psi$};
		\node [style=none] (57) at (0.75, 0.75) {};
		\node [style=none] (58) at (6.625, 0.75) {};
		\node [style=none] (59) at (0.75, -0.75) {};
		\node [style=none] (60) at (6.625, -0.75) {};
		\node [style=none] (61) at (1.25, 0.25) {$\vdots$};
		\node [style=none] (62) at (6.125, 0.25) {$\vdots$};
		\node [style=wide left state] (63) at (0.425, -2.25) {$c$};
		\node [style=none] (64) at (0.75, -1.5) {};
		\node [style=none] (65) at (0.75, -3) {};
		\node [style=none] (66) at (1.25, -2) {$\vdots$};
		\node [style=none] (67) at (6.625, -1.5) {};
		\node [style=none] (68) at (6.625, -3) {};
		\node [style=none] (69) at (6.125, -2) {$\vdots$};
		\node [style=wide right state] (70) at (7, -2.25) {$c$};
		\node [style=none] (71) at (8.5, -1) {=};
		\node [style=wide left state] (72) at (10.075, 0) {$\psi$};
		\node [style=big gate, minimum height=2.2cm] (73) at (11.825, -1.125) {$C'$};
		\node [style=big gate] (74) at (13.325, 0) {$\mathcal{O}$};
		\node [style=big gate, minimum height=2.2cm] (75) at (15.325, -1.125) {$(C')^\dagger$};
		\node [style=wide right state] (76) at (17.45, 0) {$\psi$};
		\node [style=none] (77) at (10.325, 0.75) {};
		\node [style=none] (78) at (17.2, 0.75) {};
		\node [style=none] (79) at (10.325, -0.75) {};
		\node [style=none] (80) at (17.2, -0.75) {};
		\node [style=none] (81) at (10.825, 0.25) {$\vdots$};
		\node [style=none] (82) at (16.7, 0.25) {$\vdots$};
		\node [style=wide left state] (83) at (10, -2.25) {$c$};
		\node [style=none] (84) at (10.325, -1.5) {};
		\node [style=none] (85) at (10.325, -3) {};
		\node [style=none] (86) at (10.825, -2) {$\vdots$};
		\node [style=none] (87) at (17.2, -1.5) {};
		\node [style=none] (88) at (17.2, -3) {};
		\node [style=none] (89) at (16.7, -2) {$\vdots$};
		\node [style=wide right state] (90) at (17.575, -2.25) {$c$};
		\node [style=none] (92) at (13.325, -2) {$\vdots$};
		\node [style=none] (93) at (-9.425, -6.25) {=};
		\node [style=wide left state] (94) at (-6.5, -5.25) {$\psi$};
		\node [style=big gate, minimum height=2.2cm] (95) at (-4.75, -6.375) {$C'$};
		\node [style=big gate] (96) at (-3.25, -5.25) {$\mathcal{O}$};
		\node [style=big gate, minimum height=2.2cm] (97) at (-1.25, -6.375) {$(C')^\dagger$};
		\node [style=wide right state] (98) at (0.875, -5.25) {$\psi$};
		\node [style=none] (99) at (-6.25, -4.5) {};
		\node [style=none] (100) at (0.625, -4.5) {};
		\node [style=none] (101) at (-6.25, -6) {};
		\node [style=none] (102) at (0.625, -6) {};
		\node [style=none] (103) at (-5.75, -5) {$\vdots$};
		\node [style=none] (104) at (0.125, -5) {$\vdots$};
		\node [style=wide left state] (105) at (-6.575, -7.5) {$\psi_j$};
		\node [style=none] (106) at (-6.25, -6.75) {};
		\node [style=none] (107) at (-6.25, -8.25) {};
		\node [style=none] (108) at (-5.75, -7.25) {$\vdots$};
		\node [style=none] (109) at (0.625, -6.75) {};
		\node [style=none] (110) at (0.625, -8.25) {};
		\node [style=none] (111) at (0.125, -7.25) {$\vdots$};
		\node [style=wide right state] (112) at (1, -7.5) {$\psi_j$};
		\node [style=none] (113) at (-3.25, -7.25) {$\vdots$};
		\node [style=none] (114) at (-8, -6.25) {$\sum_j \lambda_j$};
		\node [style=none] (115) at (2.575, -6.25) {=};
		\node [style=wide left state] (116) at (5.5, -5.25) {$\psi$};
		\node [style=big gate, minimum height=2.4cm] (117) at (9, -6.5) {$C'$};
		\node [style=big gate] (118) at (10.5, -5.25) {$\mathcal{O}$};
		\node [style=big gate, minimum height=2.4cm] (119) at (12.5, -6.5) {$(C')^\dagger$};
		\node [style=wide right state] (120) at (17.375, -5.25) {$\psi$};
		\node [style=none] (121) at (5.75, -4.5) {};
		\node [style=none] (122) at (17.125, -4.5) {};
		\node [style=none] (123) at (5.75, -6) {};
		\node [style=none] (124) at (17.125, -6) {};
		\node [style=none] (125) at (6.25, -5) {$\vdots$};
		\node [style=none] (126) at (16.625, -5) {$\vdots$};
		\node [style=ket] (128) at (5.5, -7) {$0$};
		\node [style=ket] (129) at (5.5, -8.5) {$0$};
		\node [style=none] (130) at (6.25, -7.5) {$\vdots$};
		\node [style=bra] (131) at (17.125, -7) {$0$};
		\node [style=bra] (132) at (17.125, -8.5) {$0$};
		\node [style=none] (133) at (16.375, -7.5) {$\vdots$};
		\node [style=none] (135) at (10.5, -7.5) {$\vdots$};
		\node [style=none] (136) at (4, -6.25) {$\sum_j \lambda_j$};
		\node [style=big gate] (137) at (7.325, -7.75) {$C_j$};
		\node [style=big gate] (138) at (15.075, -7.75) {$(C_j)^\dagger$};
	\end{pgfonlayer}
	\begin{pgfonlayer}{edgelayer}
		\draw (37.center) to (38.center);
		\draw (39.center) to (40.center);
		\draw (57.center) to (58.center);
		\draw (59.center) to (60.center);
		\draw (65.center) to (68.center);
		\draw (67.center) to (64.center);
		\draw (77.center) to (78.center);
		\draw (79.center) to (80.center);
		\draw (85.center) to (88.center);
		\draw (87.center) to (84.center);
		\draw (99.center) to (100.center);
		\draw (101.center) to (102.center);
		\draw (107.center) to (110.center);
		\draw (109.center) to (106.center);
		\draw (121.center) to (122.center);
		\draw (123.center) to (124.center);
		\draw (129) to (132);
		\draw (131) to (128);
	\end{pgfonlayer}
\end{tikzpicture} $$
    Hence, setting $\ket{\psi'_j} := \ket{\psi}\otimes \ket{0\cdots 0}$, $C'_j := C'\circ (I\otimes C_j)$ and $\mathcal{O}_j := \mathcal{O}\otimes I$, we see that we can simulate the computation of an expectation value over a circuit in $\mathcal{U}$ using a set of circuits in $\mathcal{V}$. Any computation done in $\mathcal{U}$ can then also be done using circuits in $\mathcal{V}$. Hence, $\mathcal{V}$ is computationally universal.
\end{proof}

\subsection{Completeness of extensions of the ZH-calculus}\label{app:completeness}

\begin{proof}[Proof of Theorem~\ref{thm:completeness-extension}]
    First, note that the catalysis equation is well-typed: By assumption we have $a_i^2\in \Z[\frac12, a_1,\ldots, a_{i-1}]$. We know that if we have the generators \begin{tikzpicture}
	\begin{pgfonlayer}{nodelayer}
		\node [style=none] (0) at (1.5, 0) {};
		\node [style=hadamard] (1) at (0.5, 0) {$a_j$};
	\end{pgfonlayer}
	\begin{pgfonlayer}{edgelayer}
		\draw (1) to (0.center);
	\end{pgfonlayer}
\end{tikzpicture}
 for $j\leq i-1$, then the calculus can represent any matrix with entries in $\Z[\frac12, a_1,\ldots, a_{i-1}]$. So there is some way to represent the matrix of the H-box with a label of $a_i^2$. Assuming the calculus over this ring is complete, any such way to represent this matrix is equivalent, and hence would lead to an equivalent catalysis rule.
    
    We will prove by induction on $k$ with base case $k=1$. We may assume $a_1\not\in \Z[\frac12]$, since otherwise the statement is trivial by completeness of the phase-free calculus. In that case any number in $\Z[\frac12, a_1]$ can be uniquely written as $z_1+a_1z_2$ for $z_1,z_2\in\Z[\frac12]$. We can now do all the steps we described before for $a=i$, translating a diagram $D$ containing an arbitrary number of the H-box with label $a_1$ into a diagram $D'$ in the $\Z[\frac12]$ fragment which just requires a single input of the $a_1$ H-box: $D = D'[\ket{0}+a_1\ket{1}] = D'[\ket{0}] + a_1 D'[\ket{1}]$. If we then have an equality between two diagrams $D_1$ and $D_2$ in the $\Z[\frac12,a_1]$ fragment, we get $D'_1[\ket{0}] + a_1 D'_1[\ket{1}] = D'_2[\ket{0}] + a_1 D'_2[\ket{1}]$. Because each of the component diagrams only contains elements from $\Z[\frac12]$, and a decomposition $z_1+a_1z_2$ is unique, this equation can only hold if each of the two separate components are equal. We hence again get two equalities $D'_1[\ket{0}] = D'_2[\ket{0}]$ and $D'_1[\ket{1}] = D'_2[\ket{1}]$, which allows us to conclude that $D'_1 = D'_2$ with the wire left open. As $D'_1$ and $D'_2$ are diagrams in the smaller fragment for which we have completeness, we can rewrite one into the other. Plugging in the H-box with label $a_1$ then gives us a diagrammatic proof of equality.

    The induction step follows similarly: we just observe that if $a_i\not\in \Z[\frac12, a_1,\ldots, a_{i-1}]$ that there is then again a unique way to write a number in the ring $\Z[\frac12, a_1,\ldots, a_{i-1}, a_i]$ as $z_1+a_iz_2$ where $z_1,z_2\in \Z[\frac12, a_1,\ldots, a_{i-1}]$. We can then again use the catalysis and scalar introduction rule to rewrite the diagram into the form where we can use the completeness of the smaller fragment.
\end{proof}

\section{Small angle rotations, adders and catalysis}\label{app:adder}

Catalysis is not just interesting from a theoretical viewpoint, allowing you to prove the universality and completeness of certain gate sets or generators, it is also a \emph{practically} useful tool. In this section we will see how catalysis can be used to derive an efficient way to implement small angle rotations.

To do that we first need to generalise Eq.~\eqref{eq:cat-CS-T} to allow us to implement \emph{controlled}-phase gates. To see how this works it will be helpful to first write Eq.~\eqref{eq:cat-CS-T} in circuit notation:
\begin{equation}\label{eq:cat-CS-T-circ}
    \begin{tikzpicture}
	\begin{pgfonlayer}{nodelayer}
		\node [style=none] (0) at (3.25, 0.75) {};
		\node [style=none] (1) at (3.25, -0.75) {};
		\node [style=none] (2) at (7.25, -0.75) {};
		\node [style=none] (3) at (7.25, 0.75) {};
		\node [style=cnot ctrl] (4) at (5.75, 0.75) {};
		\node [style=cnot ctrl] (7) at (4, 0.75) {};
		\node [style=cnot targ] (8) at (4, -0.75) {};
		\node [style=none] (42) at (0, 0) {=};
		\node [style=none] (43) at (-1, -0.75) {};
		\node [style=none] (44) at (-1, 0.75) {};
		\node [style=none] (46) at (-3.5, -0.75) {};
		\node [style=none] (47) at (-4.5, 0.75) {};
		\node [style=gate] (48) at (5.75, -0.75) {$Z(2\alpha)$};
		\node [style=none] (49) at (2, -0.75) {$\ket{Z(\alpha)}$};
		\node [style=gate] (50) at (-2.75, 0.75) {$Z(\alpha)$};
		\node [style=none] (51) at (-4.5, -0.75) {$\ket{Z(\alpha)}$};
	\end{pgfonlayer}
	\begin{pgfonlayer}{edgelayer}
		\draw (2.center) to (1.center);
		\draw (0.center) to (3.center);
		\draw (8) to (7);
		\draw (46.center) to (43.center);
		\draw (4) to (48);
		\draw [in=180, out=0] (47.center) to (44.center);
	\end{pgfonlayer}
\end{tikzpicture}
\end{equation}
Here we wrote a slightly more general circuit where we replace the $T$ and controlled-$S$ gates with $Z(\alpha)$ and controlled-$Z(2\alpha)$ gates. As a shorthand we write $\ket{Z(\alpha)} := Z(\alpha)\ket{+}$ as a generalisation of $\ket{T} = T\ket{+}$. 
Since this is a circuit equality that holds on the nose (with a correct global phase), it should continue to hold when we add additional control wires:
\begin{equation}\label{eq:cat-controlled-phase}
    \begin{tikzpicture}
	\begin{pgfonlayer}{nodelayer}
		\node [style=none] (0) at (3.5, 0) {};
		\node [style=none] (1) at (3.5, -1.5) {};
		\node [style=none] (2) at (7.5, -1.5) {};
		\node [style=none] (3) at (7.5, 0) {};
		\node [style=cnot ctrl] (4) at (6, 0) {};
		\node [style=cnot ctrl] (5) at (4.25, 0) {};
		\node [style=cnot targ] (6) at (4.25, -1.5) {};
		\node [style=none] (7) at (0, 0) {=};
		\node [style=none] (8) at (-1.25, -1.5) {};
		\node [style=none] (9) at (-1.25, 0) {};
		\node [style=none] (10) at (-3.75, -1.5) {};
		\node [style=none] (11) at (-4.75, 0) {};
		\node [style=gate] (12) at (6, -1.5) {$Z(2\alpha)$};
		\node [style=none] (13) at (2.25, -1.5) {$\ket{Z(\alpha)}$};
		\node [style=gate] (14) at (-3, 0) {$Z(\alpha)$};
		\node [style=none] (15) at (-4.75, -1.5) {$\ket{Z(\alpha)}$};
		\node [style=none] (16) at (-4.75, 1) {};
		\node [style=none] (18) at (-1.25, 1) {};
		\node [style=cnot ctrl] (19) at (-3, 1) {};
		\node [style=none] (20) at (-4.75, 2.25) {};
		\node [style=none] (21) at (-1.25, 2.25) {};
		\node [style=cnot ctrl] (22) at (-3, 2.25) {};
		\node [style=none] (23) at (-4, 1.875) {$\vdots$};
		\node [style=none] (24) at (3.5, 1) {};
		\node [style=none] (25) at (7.5, 1) {};
		\node [style=cnot ctrl] (26) at (4.25, 1) {};
		\node [style=none] (27) at (3.5, 2.25) {};
		\node [style=none] (28) at (7.5, 2.25) {};
		\node [style=cnot ctrl] (29) at (4.25, 2.25) {};
		\node [style=none] (30) at (5.125, 1.875) {$\vdots$};
		\node [style=cnot ctrl] (31) at (6, 1) {};
		\node [style=cnot ctrl] (32) at (6, 2.25) {};
	\end{pgfonlayer}
	\begin{pgfonlayer}{edgelayer}
		\draw (2.center) to (1.center);
		\draw (0.center) to (3.center);
		\draw (6) to (5);
		\draw (10.center) to (8.center);
		\draw (4) to (12);
		\draw [in=180, out=0] (11.center) to (9.center);
		\draw (16.center) to (18.center);
		\draw (20.center) to (21.center);
		\draw (22) to (14);
		\draw (24.center) to (25.center);
		\draw (27.center) to (28.center);
		\draw (29) to (5);
		\draw (32) to (4);
	\end{pgfonlayer}
\end{tikzpicture}
\end{equation}
We can prove this is correct using ZH:
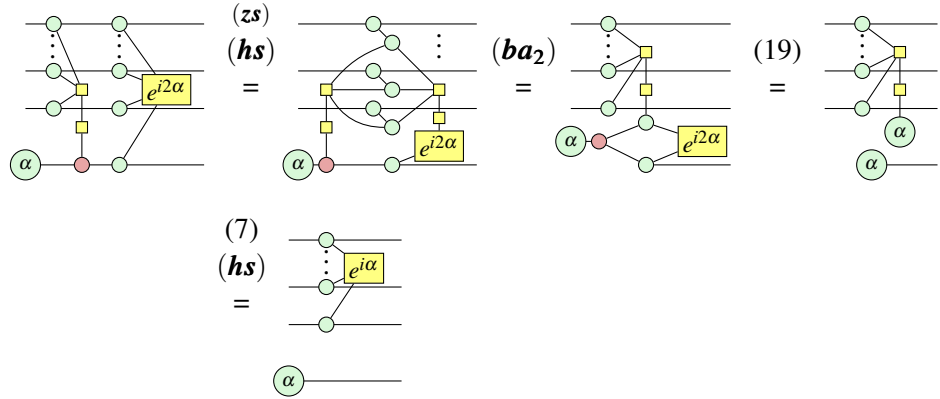
\begin{equation*}
    \begin{tikzpicture}
	\begin{pgfonlayer}{nodelayer}
		\node [style=none] (0) at (-5.75, -0.5) {};
		\node [style=Z phase dot] (1) at (-5.75, -2) {$\alpha$};
		\node [style=none] (2) at (-1, -2) {};
		\node [style=none] (3) at (-1, -0.5) {};
		\node [style=Z dot] (4) at (-3.25, -0.5) {};
		\node [style=Z dot] (5) at (-5, -0.5) {};
		\node [style=X dot] (6) at (-4.25, -2) {};
		\node [style=none] (7) at (0.25, 0) {=};
		\node [style=none] (23) at (-5.75, 0.5) {};
		\node [style=none] (24) at (-1, 0.5) {};
		\node [style=Z dot] (25) at (-5, 0.5) {};
		\node [style=none] (26) at (-5.75, 1.75) {};
		\node [style=none] (27) at (-1, 1.75) {};
		\node [style=Z dot] (28) at (-5, 1.75) {};
		\node [style=none] (29) at (-5.025, 1.375) {$\vdots$};
		\node [style=Z dot] (30) at (-3.25, 0.5) {};
		\node [style=Z dot] (31) at (-3.25, 1.75) {};
		\node [style=hadamard] (32) at (-4.25, 0) {};
		\node [style=hadamard] (33) at (-4.25, -1) {};
		\node [style=hadamard] (34) at (-2, 0) {$e^{i2\alpha}$};
		\node [style=Z dot] (35) at (-3.25, -2) {};
		\node [style=none] (36) at (-3.25, 1.375) {$\vdots$};
		\node [style=none] (37) at (1.5, -0.5) {};
		\node [style=Z phase dot] (38) at (1.5, -2) {$\alpha$};
		\node [style=none] (39) at (6.25, -2) {};
		\node [style=none] (40) at (6.25, -0.5) {};
		\node [style=Z dot] (41) at (3.5, -0.5) {};
		\node [style=X dot] (43) at (2.25, -2) {};
		\node [style=none] (44) at (1.5, 0.5) {};
		\node [style=none] (45) at (6.25, 0.5) {};
		\node [style=none] (47) at (1.5, 1.75) {};
		\node [style=none] (48) at (6.25, 1.75) {};
		\node [style=Z dot] (49) at (3.5, 1.75) {};
		\node [style=Z dot] (51) at (3.5, 0.5) {};
		\node [style=hadamard] (53) at (2.25, 0) {};
		\node [style=hadamard] (54) at (2.25, -1) {};
		\node [style=hadamard] (55) at (5.25, 0) {};
		\node [style=Z dot] (56) at (4, -2) {};
		\node [style=none] (57) at (5.25, 1.375) {$\vdots$};
		\node [style=none] (58) at (0.25, 1) {\HFuseRule};
		\node [style=none] (59) at (0.25, 2) {\spiderrule};
		\node [style=hadamard] (60) at (5.25, -1.5) {$e^{i2\alpha}$};
		\node [style=hadamard] (61) at (5.25, -0.75) {};
		\node [style=Z dot] (62) at (4, 1.25) {};
		\node [style=Z dot] (63) at (4, 0) {};
		\node [style=Z dot] (64) at (4, -1) {};
		\node [style=none] (65) at (7.5, 0) {=};
		\node [style=none] (66) at (8.75, -0.5) {};
		\node [style=Z phase dot] (67) at (8.75, -1.375) {$\alpha$};
		\node [style=none] (68) at (13, -2) {};
		\node [style=none] (69) at (13, -0.5) {};
		\node [style=Z dot] (70) at (9.75, -0.5) {};
		\node [style=X dot] (71) at (9.5, -1.375) {};
		\node [style=none] (72) at (8.75, 0.5) {};
		\node [style=none] (73) at (13, 0.5) {};
		\node [style=none] (74) at (8.75, 1.75) {};
		\node [style=none] (75) at (13, 1.75) {};
		\node [style=Z dot] (76) at (9.75, 1.75) {};
		\node [style=Z dot] (77) at (9.75, 0.5) {};
		\node [style=Z dot] (81) at (10.75, -2) {};
		\node [style=none] (82) at (9.75, 1.375) {$\vdots$};
		\node [style=none] (84) at (7.5, 1) {\HCompRule};
		\node [style=hadamard] (85) at (12.25, -1.375) {$e^{i2\alpha}$};
		\node [style=hadamard] (87) at (10.75, 1) {};
		\node [style=hadamard] (88) at (10.75, 0) {};
		\node [style=Z dot] (89) at (10.75, -0.875) {};
		\node [style=none] (90) at (14.25, 0) {=};
		\node [style=none] (91) at (15.5, -0.5) {};
		\node [style=Z phase dot] (92) at (16.75, -2) {$\alpha$};
		\node [style=none] (93) at (18.5, -2) {};
		\node [style=none] (94) at (18.5, -0.5) {};
		\node [style=Z dot] (95) at (16.5, -0.5) {};
		\node [style=none] (97) at (15.5, 0.5) {};
		\node [style=none] (98) at (18.5, 0.5) {};
		\node [style=none] (99) at (15.5, 1.75) {};
		\node [style=none] (100) at (18.5, 1.75) {};
		\node [style=Z dot] (101) at (16.5, 1.75) {};
		\node [style=Z dot] (102) at (16.5, 0.5) {};
		\node [style=none] (104) at (16.5, 1.375) {$\vdots$};
		\node [style=none] (105) at (14.25, 1) {\eqref{eq:H-box-catalysis}};
		\node [style=hadamard] (107) at (17.5, 1) {};
		\node [style=hadamard] (108) at (17.5, 0) {};
		\node [style=Z phase dot] (109) at (17.5, -1.125) {$\alpha$};
		\node [style=none] (110) at (0, -5.75) {=};
		\node [style=none] (111) at (1.25, -6.25) {};
		\node [style=Z phase dot] (112) at (1.25, -7.75) {$\alpha$};
		\node [style=none] (113) at (4.25, -7.75) {};
		\node [style=none] (114) at (4.25, -6.25) {};
		\node [style=Z dot] (115) at (2.25, -6.25) {};
		\node [style=none] (116) at (1.25, -5.25) {};
		\node [style=none] (117) at (4.25, -5.25) {};
		\node [style=none] (118) at (1.25, -4) {};
		\node [style=none] (119) at (4.25, -4) {};
		\node [style=Z dot] (120) at (2.25, -4) {};
		\node [style=Z dot] (121) at (2.25, -5.25) {};
		\node [style=none] (122) at (2.25, -4.375) {$\vdots$};
		\node [style=none] (123) at (0, -3.75) {\eqref{eq:Z-a-state}};
		\node [style=hadamard] (124) at (3.25, -4.75) {$e^{i\alpha}$};
		\node [style=none] (127) at (0, -4.75) {\HFuseRule};
	\end{pgfonlayer}
	\begin{pgfonlayer}{edgelayer}
		\draw (2.center) to (1);
		\draw (0.center) to (3.center);
		\draw (23.center) to (24.center);
		\draw (26.center) to (27.center);
		\draw (28) to (32);
		\draw (25) to (32);
		\draw (5) to (32);
		\draw (32) to (33);
		\draw (33) to (6);
		\draw (35) to (34);
		\draw (4) to (34);
		\draw (30) to (34);
		\draw (31) to (34);
		\draw (39.center) to (38);
		\draw (37.center) to (40.center);
		\draw (44.center) to (45.center);
		\draw (47.center) to (48.center);
		\draw (53) to (54);
		\draw (54) to (43);
		\draw (56) to (60);
		\draw (60) to (61);
		\draw (61) to (55);
		\draw [bend left=15] (53) to (62);
		\draw (62) to (55);
		\draw (49) to (62);
		\draw (63) to (55);
		\draw (63) to (53);
		\draw (51) to (63);
		\draw [bend left] (64) to (53);
		\draw (64) to (41);
		\draw (64) to (55);
		\draw (66.center) to (69.center);
		\draw (72.center) to (73.center);
		\draw (74.center) to (75.center);
		\draw (81) to (85);
		\draw (76) to (87);
		\draw (77) to (87);
		\draw (70) to (87);
		\draw (87) to (88);
		\draw (88) to (89);
		\draw (89) to (85);
		\draw (89) to (71);
		\draw (67) to (71);
		\draw (71) to (81);
		\draw (81) to (68.center);
		\draw (91.center) to (94.center);
		\draw (97.center) to (98.center);
		\draw (99.center) to (100.center);
		\draw (101) to (107);
		\draw (102) to (107);
		\draw (95) to (107);
		\draw (107) to (108);
		\draw (108) to (109);
		\draw (92) to (93.center);
		\draw (111.center) to (114.center);
		\draw (116.center) to (117.center);
		\draw (118.center) to (119.center);
		\draw (120) to (124);
		\draw (121) to (124);
		\draw (115) to (124);
		\draw (112) to (113.center);
	\end{pgfonlayer}
\end{tikzpicture}
\end{equation*}
Because we can apply catalysis equally well to controlled phases, we can start iterating the procedure producing bigger and bigger controlled-phase gates, where the phase being controlled is also increasingly large. For instance, if we want to implement a $T$ gate, we can do the following:
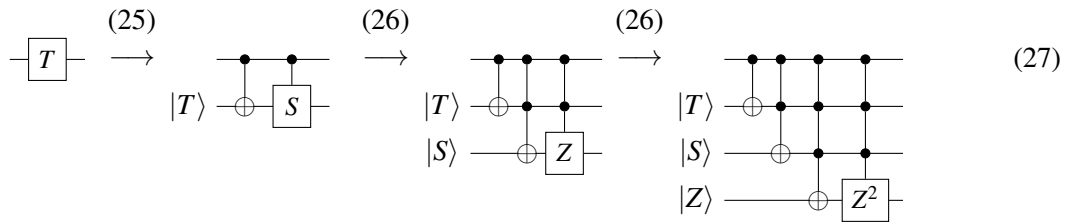
\begin{equation}\label{eq:cat-T-repeated-circ}
    \begin{tikzpicture}
	\begin{pgfonlayer}{nodelayer}
		\node [style=none] (0) at (2.25, 0) {};
		\node [style=none] (1) at (2.25, -1.25) {};
		\node [style=none] (2) at (5.25, -1.25) {};
		\node [style=none] (3) at (5.25, 0) {};
		\node [style=cnot ctrl] (4) at (4.25, 0) {};
		\node [style=cnot ctrl] (5) at (3, 0) {};
		\node [style=cnot targ] (6) at (3, -1.25) {};
		\node [style=none] (7) at (0, 0) {$\longrightarrow$};
		\node [style=none] (9) at (-1.25, 0) {};
		\node [style=none] (11) at (-3.25, 0) {};
		\node [style=gate] (12) at (4.25, -1.25) {$S$};
		\node [style=none] (13) at (1.5, -1.25) {$\ket{T}$};
		\node [style=gate] (14) at (-2.25, 0) {$T$};
		\node [style=none] (15) at (0, 1) {\eqref{eq:cat-CS-T-circ}};
		\node [style=none] (16) at (9, 0) {};
		\node [style=none] (17) at (9, -1.25) {};
		\node [style=none] (18) at (12.5, -1.25) {};
		\node [style=none] (19) at (12.5, 0) {};
		\node [style=cnot ctrl] (20) at (11.5, 0) {};
		\node [style=cnot ctrl] (21) at (9.75, 0) {};
		\node [style=cnot targ] (22) at (9.75, -1.25) {};
		\node [style=none] (23) at (6.75, 0) {$\longrightarrow$};
		\node [style=none] (25) at (8.25, -1.25) {$\ket{T}$};
		\node [style=none] (26) at (6.75, 1) {\eqref{eq:cat-controlled-phase}};
		\node [style=none] (27) at (9, -2.5) {};
		\node [style=none] (28) at (12.5, -2.5) {};
		\node [style=gate] (30) at (11.5, -2.5) {$Z$};
		\node [style=none] (31) at (8.25, -2.5) {$\ket{S}$};
		\node [style=cnot ctrl] (32) at (11.5, -1.25) {};
		\node [style=cnot ctrl] (33) at (10.5, 0) {};
		\node [style=cnot targ] (34) at (10.5, -2.5) {};
		\node [style=cnot ctrl] (35) at (10.5, -1.25) {};
		\node [style=none] (36) at (15.75, 0) {};
		\node [style=none] (37) at (15.75, -1.25) {};
		\node [style=none] (38) at (20.5, -1.25) {};
		\node [style=none] (39) at (20.5, 0) {};
		\node [style=cnot ctrl] (40) at (19.5, 0) {};
		\node [style=cnot ctrl] (41) at (16.5, 0) {};
		\node [style=cnot targ] (42) at (16.5, -1.25) {};
		\node [style=none] (43) at (13.5, 0) {$\longrightarrow$};
		\node [style=none] (44) at (15, -1.25) {$\ket{T}$};
		\node [style=none] (45) at (13.5, 1) {\eqref{eq:cat-controlled-phase}};
		\node [style=none] (46) at (15.75, -2.5) {};
		\node [style=none] (47) at (20.5, -2.5) {};
		\node [style=none] (49) at (15, -2.5) {$\ket{S}$};
		\node [style=cnot ctrl] (50) at (19.5, -1.25) {};
		\node [style=cnot ctrl] (51) at (17.25, 0) {};
		\node [style=cnot targ] (52) at (17.25, -2.5) {};
		\node [style=cnot ctrl] (53) at (17.25, -1.25) {};
		\node [style=none] (54) at (15.75, -3.75) {};
		\node [style=none] (55) at (20.5, -3.75) {};
		\node [style=gate] (56) at (19.5, -3.75) {$Z^2$};
		\node [style=none] (57) at (15, -3.75) {$\ket{Z}$};
		\node [style=cnot ctrl] (59) at (19.5, -2.5) {};
		\node [style=cnot ctrl] (60) at (18.25, 0) {};
		\node [style=cnot targ] (61) at (18.25, -3.75) {};
		\node [style=cnot ctrl] (62) at (18.25, -1.25) {};
		\node [style=cnot ctrl] (63) at (18.25, -2.5) {};
	\end{pgfonlayer}
	\begin{pgfonlayer}{edgelayer}
		\draw (2.center) to (1.center);
		\draw (0.center) to (3.center);
		\draw (6) to (5);
		\draw (4) to (12);
		\draw [in=180, out=0] (11.center) to (9.center);
		\draw (18.center) to (17.center);
		\draw (16.center) to (19.center);
		\draw (22) to (21);
		\draw (28.center) to (27.center);
		\draw (20) to (32);
		\draw (32) to (30);
		\draw (33) to (35);
		\draw (35) to (34);
		\draw (38.center) to (37.center);
		\draw (36.center) to (39.center);
		\draw (42) to (41);
		\draw (47.center) to (46.center);
		\draw (40) to (50);
		\draw (51) to (53);
		\draw (53) to (52);
		\draw (55.center) to (54.center);
		\draw (50) to (56);
		\draw (60) to (62);
		\draw (62) to (61);
	\end{pgfonlayer}
\end{tikzpicture}
\end{equation}
Here in the last step we are left with a controlled $Z^2$ operation. But since $Z^2=\text{id}$ this does not do anything and we can remove it. So at this point we can stop the iteration of the catalysis. We see then that we can implement a $T$ gate just using multiple-controlled Toffoli gates, if we have the right catalysis states lying around. This procedure works to implement any $Z(2\pi/2^k)$ gate: we then get a ladder of $k$ Toffoli gates.
This implements a controlled-decrementer circuit that decreases the value of an $n$-bit number by $1$, controlled on the top wire. By making a ladder of these controlled-decrementers we implement a subtraction circuit that maps $\ket{a,b}\mapsto \ket{a,b-a}$ for some $n$-bit numbers $a$ and $b$. For this reason, when we apply a subtraction circuit to a collection of catalysis states, this implements phase gates on on the top qubits:
\begin{equation}\label{eq:cat-sub-circ}
    \begin{tikzpicture}
	\begin{pgfonlayer}{nodelayer}
		\node [style=none] (24) at (-6.75, 0.5) {};
		\node [style=none] (25) at (-6.75, -0.5) {};
		\node [style=none] (26) at (-4.25, -0.5) {};
		\node [style=none] (27) at (-4.25, 0.5) {};
		\node [style=none] (31) at (-7.5, -0.5) {$\ket{T}$};
		\node [style=none] (32) at (-6.75, -1.5) {};
		\node [style=none] (33) at (-4.25, -1.5) {};
		\node [style=none] (34) at (-7.5, -1.5) {$\ket{S}$};
		\node [style=none] (39) at (-6.75, -2.5) {};
		\node [style=none] (40) at (-4.25, -2.5) {};
		\node [style=none] (42) at (-7.5, -2.5) {$\ket{Z}$};
		\node [style=none] (48) at (-6.75, 1.75) {};
		\node [style=none] (49) at (-4.25, 1.75) {};
		\node [style=none] (50) at (-6.75, 3) {};
		\node [style=none] (51) at (-4.25, 3) {};
		\node [style=large box, minimum height=3.2cm] (52) at (-5.5, 0.25) {Sub};
		\node [style=large box, minimum height=1.4cm] (53) at (0.75, -1.475) {$-1$};
		\node [style=cnot ctrl] (54) at (0.75, 3) {};
		\node [style=large box, minimum height=1.0cm] (55) at (2.25, -1.875) {$-1$};
		\node [style=cnot ctrl] (56) at (2.25, 1.75) {};
		\node [style=box] (57) at (3.75, -2.5) {$-1$};
		\node [style=cnot ctrl] (58) at (3.75, 0.5) {};
		\node [style=none] (60) at (-0.5, 0.5) {};
		\node [style=none] (61) at (-0.5, -0.5) {};
		\node [style=none] (62) at (4.75, -0.5) {};
		\node [style=none] (63) at (4.75, 0.5) {};
		\node [style=none] (64) at (-1.25, -0.5) {$\ket{T}$};
		\node [style=none] (65) at (-0.5, -1.5) {};
		\node [style=none] (66) at (4.75, -1.5) {};
		\node [style=none] (67) at (-1.25, -1.5) {$\ket{S}$};
		\node [style=none] (68) at (-0.5, -2.5) {};
		\node [style=none] (69) at (4.75, -2.5) {};
		\node [style=none] (70) at (-1.25, -2.5) {$\ket{Z}$};
		\node [style=none] (71) at (-0.5, 1.75) {};
		\node [style=none] (72) at (4.75, 1.75) {};
		\node [style=none] (73) at (-0.5, 3) {};
		\node [style=none] (74) at (4.75, 3) {};
		\node [style=none] (76) at (-2.75, 0) {:=};
		\node [style=large box, minimum height=1.0cm] (80) at (10.25, -1.875) {$-1$};
		\node [style=cnot ctrl] (81) at (10.25, 1.75) {};
		\node [style=box] (82) at (11.75, -2.5) {$-1$};
		\node [style=cnot ctrl] (83) at (11.75, 0.5) {};
		\node [style=none] (84) at (8.5, 0.5) {};
		\node [style=none] (85) at (8.5, -0.5) {};
		\node [style=none] (86) at (12.75, -0.5) {};
		\node [style=none] (87) at (12.75, 0.5) {};
		\node [style=none] (88) at (7.75, -0.5) {$\ket{T}$};
		\node [style=none] (89) at (8.5, -1.5) {};
		\node [style=none] (90) at (12.75, -1.5) {};
		\node [style=none] (91) at (7.75, -1.5) {$\ket{S}$};
		\node [style=none] (92) at (8.5, -2.5) {};
		\node [style=none] (93) at (12.75, -2.5) {};
		\node [style=none] (94) at (7.75, -2.5) {$\ket{Z}$};
		\node [style=none] (95) at (8.5, 1.75) {};
		\node [style=none] (96) at (12.75, 1.75) {};
		\node [style=none] (97) at (8.5, 3) {};
		\node [style=none] (98) at (12.75, 3) {};
		\node [style=none] (99) at (6, 0) {=};
		\node [style=none] (100) at (6, 1) {\eqref{eq:cat-T-repeated-circ}};
		\node [style=gate] (101) at (9.5, 3) {$T$};
		\node [style=box] (104) at (0.75, -10.25) {$-1$};
		\node [style=cnot ctrl] (105) at (0.75, -7.25) {};
		\node [style=none] (106) at (-0.25, -7.25) {};
		\node [style=none] (107) at (-0.25, -8.25) {};
		\node [style=none] (108) at (1.75, -8.25) {};
		\node [style=none] (109) at (1.75, -7.25) {};
		\node [style=none] (110) at (-1, -8.25) {$\ket{T}$};
		\node [style=none] (111) at (-0.25, -9.25) {};
		\node [style=none] (112) at (1.75, -9.25) {};
		\node [style=none] (113) at (-1, -9.25) {$\ket{S}$};
		\node [style=none] (114) at (-0.25, -10.25) {};
		\node [style=none] (115) at (1.75, -10.25) {};
		\node [style=none] (116) at (-1, -10.25) {$\ket{Z}$};
		\node [style=none] (117) at (-0.25, -6) {};
		\node [style=none] (118) at (1.75, -6) {};
		\node [style=none] (119) at (-0.25, -4.75) {};
		\node [style=none] (120) at (1.75, -4.75) {};
		\node [style=none] (121) at (-2.75, -7.75) {=};
		\node [style=none] (122) at (-2.75, -6.75) {\eqref{eq:cat-T-repeated-circ}};
		\node [style=gate] (123) at (0.75, -4.75) {$T$};
		\node [style=gate] (124) at (0.75, -6) {$S$};
		\node [style=none] (127) at (5.5, -7.25) {};
		\node [style=none] (128) at (5.5, -8.25) {};
		\node [style=none] (129) at (7.5, -8.25) {};
		\node [style=none] (130) at (7.5, -7.25) {};
		\node [style=none] (131) at (4.75, -8.25) {$\ket{T}$};
		\node [style=none] (132) at (5.5, -9.25) {};
		\node [style=none] (133) at (7.5, -9.25) {};
		\node [style=none] (134) at (4.75, -9.25) {$\ket{S}$};
		\node [style=none] (135) at (5.5, -10.25) {};
		\node [style=none] (136) at (7.5, -10.25) {};
		\node [style=none] (137) at (4.75, -10.25) {$\ket{Z}$};
		\node [style=none] (138) at (5.5, -6) {};
		\node [style=none] (139) at (7.5, -6) {};
		\node [style=none] (140) at (5.5, -4.75) {};
		\node [style=none] (141) at (7.5, -4.75) {};
		\node [style=none] (142) at (3, -7.75) {=};
		\node [style=none] (143) at (3, -6.75) {\eqref{eq:cat-T-repeated-circ}};
		\node [style=gate] (144) at (6.5, -4.75) {$T$};
		\node [style=gate] (145) at (6.5, -6) {$S$};
		\node [style=gate] (146) at (6.5, -7.25) {$Z$};
	\end{pgfonlayer}
	\begin{pgfonlayer}{edgelayer}
		\draw (26.center) to (25.center);
		\draw (24.center) to (27.center);
		\draw (33.center) to (32.center);
		\draw (40.center) to (39.center);
		\draw (48.center) to (49.center);
		\draw (50.center) to (51.center);
		\draw (54) to (53);
		\draw (56) to (55);
		\draw (58) to (57);
		\draw (62.center) to (61.center);
		\draw (60.center) to (63.center);
		\draw (66.center) to (65.center);
		\draw (69.center) to (68.center);
		\draw (71.center) to (72.center);
		\draw (73.center) to (74.center);
		\draw (81) to (80);
		\draw (83) to (82);
		\draw (86.center) to (85.center);
		\draw (84.center) to (87.center);
		\draw (90.center) to (89.center);
		\draw (93.center) to (92.center);
		\draw (95.center) to (96.center);
		\draw (97.center) to (98.center);
		\draw (105) to (104);
		\draw (108.center) to (107.center);
		\draw (106.center) to (109.center);
		\draw (112.center) to (111.center);
		\draw (115.center) to (114.center);
		\draw (117.center) to (118.center);
		\draw (119.center) to (120.center);
		\draw (129.center) to (128.center);
		\draw (127.center) to (130.center);
		\draw (133.center) to (132.center);
		\draw (136.center) to (135.center);
		\draw (138.center) to (139.center);
		\draw (140.center) to (141.center);
	\end{pgfonlayer}
\end{tikzpicture}
\end{equation}
An adder can be implemented quite efficiently, so we transform Eq.~\eqref{eq:cat-sub-circ} slightly, so that it uses an adder instead of a subtracter. By taking Eq.~\eqref{eq:cat-sub-circ} and composing both sides on the right by $\text{Add} = \text{Sub}^\dagger$, and on the left by $(T\otimes S\otimes Z)^\dagger$. After cancelling with the adjoints we are then left with the following equation:
\begin{equation}\label{eq:cat-addition-circ}
    \begin{tikzpicture}
	\begin{pgfonlayer}{nodelayer}
		\node [style=none] (0) at (-4, 0.5) {};
		\node [style=none] (1) at (-4, -0.5) {};
		\node [style=none] (2) at (-1.5, -0.5) {};
		\node [style=none] (3) at (-1.5, 0.5) {};
		\node [style=none] (4) at (-4.75, -0.5) {$\ket{T}$};
		\node [style=none] (5) at (-4, -1.5) {};
		\node [style=none] (6) at (-1.5, -1.5) {};
		\node [style=none] (7) at (-4.75, -1.5) {$\ket{S}$};
		\node [style=none] (8) at (-4, -2.5) {};
		\node [style=none] (9) at (-1.5, -2.5) {};
		\node [style=none] (10) at (-4.75, -2.5) {$\ket{Z}$};
		\node [style=none] (11) at (-4, 1.75) {};
		\node [style=none] (12) at (-1.5, 1.75) {};
		\node [style=none] (13) at (-4, 3) {};
		\node [style=none] (14) at (-1.5, 3) {};
		\node [style=large box, minimum height=3.2cm] (15) at (3.25, 0.25) {Add};
		\node [style=none] (37) at (0, 0) {=};
		\node [style=none] (82) at (2, 0.5) {};
		\node [style=none] (83) at (2, -0.5) {};
		\node [style=none] (84) at (4.5, -0.5) {};
		\node [style=none] (85) at (4.5, 0.5) {};
		\node [style=none] (86) at (1.25, -0.5) {$\ket{T}$};
		\node [style=none] (87) at (2, -1.5) {};
		\node [style=none] (88) at (4.5, -1.5) {};
		\node [style=none] (89) at (1.25, -1.5) {$\ket{S}$};
		\node [style=none] (90) at (2, -2.5) {};
		\node [style=none] (91) at (4.5, -2.5) {};
		\node [style=none] (92) at (1.25, -2.5) {$\ket{Z}$};
		\node [style=none] (93) at (2, 1.75) {};
		\node [style=none] (94) at (4.5, 1.75) {};
		\node [style=none] (95) at (2, 3) {};
		\node [style=none] (96) at (4.5, 3) {};
		\node [style=gate] (99) at (-2.75, 3) {$T^\dagger$};
		\node [style=gate] (100) at (-2.75, 1.75) {$S^\dagger$};
		\node [style=gate] (101) at (-2.75, 0.5) {$Z^\dagger$};
	\end{pgfonlayer}
	\begin{pgfonlayer}{edgelayer}
		\draw (2.center) to (1.center);
		\draw (0.center) to (3.center);
		\draw (6.center) to (5.center);
		\draw (9.center) to (8.center);
		\draw (11.center) to (12.center);
		\draw (13.center) to (14.center);
		\draw (84.center) to (83.center);
		\draw (82.center) to (85.center);
		\draw (88.center) to (87.center);
		\draw (91.center) to (90.center);
		\draw (93.center) to (94.center);
		\draw (95.center) to (96.center);
	\end{pgfonlayer}
\end{tikzpicture}
\end{equation}
We showed the construction here for 3 bits, but this works for any number of bits $n$, in which case the smallest phase we implement is $Z(2\pi/2^n)$.

Such a series of parallel phases is not that useful, but by using ancillae we can make them work on the same qubit.
First, we can transfer the application of a phase gate to a zeroed ancilla:
\begin{equation}
    \begin{tikzpicture}
	\begin{pgfonlayer}{nodelayer}
		\node [style=none] (0) at (-5, 0.75) {};
		\node [style=X dot] (1) at (-5, -0.75) {};
		\node [style=X dot] (2) at (-1, -0.75) {};
		\node [style=X dot] (3) at (-4, -0.75) {};
		\node [style=X dot] (4) at (-2, -0.75) {};
		\node [style=Z phase dot] (5) at (-3, -0.75) {$\alpha$};
		\node [style=Z dot] (6) at (-4, 0.75) {};
		\node [style=Z dot] (7) at (-2, 0.75) {};
		\node [style=none] (8) at (-1, 0.75) {};
		\node [style=none] (9) at (0, 0) {=};
		\node [style=none] (10) at (1, 0.5) {};
		\node [style=X dot] (13) at (2, -0.75) {};
		\node [style=X dot] (14) at (4, -0.75) {};
		\node [style=Z phase dot] (15) at (3, -0.75) {$\alpha$};
		\node [style=Z dot] (16) at (2, 0.5) {};
		\node [style=Z dot] (17) at (4, 0.5) {};
		\node [style=none] (18) at (5, 0.5) {};
		\node [style=none] (19) at (0, 1) {\spiderrule};
		\node [style=none] (20) at (6, 0) {=};
		\node [style=none] (21) at (7, 0.25) {};
		\node [style=Z phase dot] (24) at (9, -0.75) {$\alpha$};
		\node [style=Z dot] (25) at (8, 0.25) {};
		\node [style=Z dot] (26) at (10, 0.25) {};
		\node [style=none] (27) at (11, 0.25) {};
		\node [style=none] (28) at (6, 1) {\idrule};
		\node [style=none] (29) at (12, 0) {=};
		\node [style=none] (30) at (13, 0) {};
		\node [style=Z phase dot] (31) at (14, 0) {$\alpha$};
		\node [style=none] (34) at (15, 0) {};
		\node [style=none] (35) at (12, 1) {\spiderrule};
	\end{pgfonlayer}
	\begin{pgfonlayer}{edgelayer}
		\draw (0.center) to (6);
		\draw (6) to (3);
		\draw (1) to (2);
		\draw (6) to (7);
		\draw (7) to (4);
		\draw (7) to (8.center);
		\draw (10.center) to (16);
		\draw (16) to (13);
		\draw (16) to (17);
		\draw (17) to (14);
		\draw (17) to (18.center);
		\draw (13) to (14);
		\draw (21.center) to (25);
		\draw (25) to (26);
		\draw (26) to (27.center);
		\draw (25) to (24);
		\draw (24) to (26);
		\draw (30.center) to (34.center);
	\end{pgfonlayer}
\end{tikzpicture}
\end{equation}
Now when we have a complicated phase, we can decompose it into its components, and put each of these onto its own ancilla. Suppose for instance we want to implement the phase $Z(\frac{11}{8}\pi)$. We can then write $11$ bitwise as $1011$ so that $Z(\frac{11}{8}\pi) = Z(2\pi/2^4 (2^3 + 2^1 + 2^0))$. We can then put each of these component phases onto their own ancilla to get:
\begin{equation}\label{eq:phase-gate-on-ancillae}
    \begin{tikzpicture}
	\begin{pgfonlayer}{nodelayer}
		\node [style=none] (0) at (7, 0) {};
		\node [style=X dot] (1) at (7, -1) {};
		\node [style=X dot] (2) at (14.25, -1) {};
		\node [style=X dot] (3) at (8, -1) {};
		\node [style=Z phase dot] (5) at (10.5, -1) {$\frac\pi8$};
		\node [style=Z dot] (6) at (8, 0) {};
		\node [style=none] (8) at (14.25, 0) {};
		\node [style=none] (9) at (0, 0) {=};
		\node [style=none] (26) at (-3.75, 0) {};
		\node [style=Z phase dot] (27) at (-2.5, 0) {$\frac{11}{8}\pi$};
		\node [style=none] (28) at (-1.25, 0) {};
		\node [style=none] (29) at (1, 0) {};
		\node [style=Z phase dot] (30) at (2, 0) {$\pi$};
		\node [style=none] (31) at (5, 0) {};
		\node [style=Z phase dot] (32) at (3, 0) {$\frac\pi4$};
		\node [style=Z phase dot] (33) at (4, 0) {$\frac\pi8$};
		\node [style=none] (34) at (6, 0) {=};
		\node [style=X dot] (35) at (7, -2) {};
		\node [style=X dot] (36) at (14.25, -2) {};
		\node [style=X dot] (37) at (8.75, -2) {};
		\node [style=Z phase dot] (39) at (10.5, -2) {$\frac\pi4$};
		\node [style=X dot] (40) at (7, -3) {};
		\node [style=X dot] (41) at (14.25, -3) {};
		\node [style=Z phase dot] (44) at (10.5, -3) {$\frac\pi2$};
		\node [style=X dot] (45) at (7, -4) {};
		\node [style=X dot] (46) at (14.25, -4) {};
		\node [style=X dot] (47) at (9.5, -4) {};
		\node [style=Z phase dot] (49) at (10.5, -4) {$\pi$};
		\node [style=Z dot] (50) at (8.75, 0) {};
		\node [style=Z dot] (51) at (9.5, 0) {};
		\node [style=X dot] (52) at (13.25, -1) {};
		\node [style=Z dot] (53) at (13.25, 0) {};
		\node [style=X dot] (54) at (12.5, -2) {};
		\node [style=X dot] (55) at (11.75, -4) {};
		\node [style=Z dot] (56) at (12.5, 0) {};
		\node [style=Z dot] (57) at (11.75, 0) {};
	\end{pgfonlayer}
	\begin{pgfonlayer}{edgelayer}
		\draw (0.center) to (6);
		\draw (6) to (3);
		\draw (1) to (2);
		\draw (26.center) to (28.center);
		\draw (29.center) to (31.center);
		\draw (35) to (36);
		\draw (40) to (41);
		\draw (45) to (46);
		\draw (50) to (37);
		\draw (47) to (51);
		\draw (53) to (52);
		\draw (56) to (54);
		\draw (55) to (57);
		\draw (8.center) to (6);
	\end{pgfonlayer}
\end{tikzpicture}
\end{equation}
We have here also added a zeroed ancilla that gets a $Z(\frac\pi2)$ applied that does nothing. We need this qubit to complete the pattern: we see then that we get the right shape needed to use Eq.~\eqref{eq:cat-addition-circ}. However, note that Eq.~\eqref{eq:cat-addition-circ} has adjoint phases, instead of the actual phases we need. There are multiple ways we can deal with this. One way is to realise that for phase gates, the adjoint is the conjugate: $T^\dagger = \overline{T}$. Hence, if we take the conjugate of both sides of Eq.~\eqref{eq:cat-addition-circ} we do get the right phases. Since the Adder is a real matrix, this stays the same, but the states needed for the catalysis also flip: $\overline{\ket{T}} = \ket{T^\dagger}$. We then have everything we need to produce the circuit we are after:
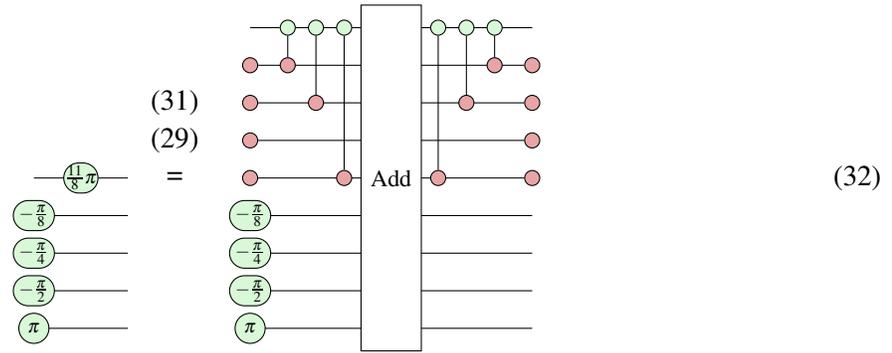
\begin{equation}\label{eq:phase-gate-via-catalysis}
    \begin{tikzpicture}
	\begin{pgfonlayer}{nodelayer}
		\node [style=none] (0) at (2, 4) {};
		\node [style=X dot] (1) at (2, 3) {};
		\node [style=X dot] (2) at (9.5, 3) {};
		\node [style=X dot] (3) at (3, 3) {};
		\node [style=Z dot] (5) at (3, 4) {};
		\node [style=none] (6) at (9.5, 4) {};
		\node [style=none] (7) at (0, 0) {=};
		\node [style=none] (8) at (-3.75, 0) {};
		\node [style=Z phase dot] (9) at (-2.5, 0) {$\frac{11}{8}\pi$};
		\node [style=none] (10) at (-1.25, 0) {};
		\node [style=X dot] (17) at (2, 2) {};
		\node [style=X dot] (18) at (9.5, 2) {};
		\node [style=X dot] (19) at (3.75, 2) {};
		\node [style=X dot] (21) at (2, 1) {};
		\node [style=X dot] (22) at (9.5, 1) {};
		\node [style=X dot] (24) at (2, 0) {};
		\node [style=X dot] (25) at (9.5, 0) {};
		\node [style=X dot] (26) at (4.5, 0) {};
		\node [style=Z dot] (28) at (3.75, 4) {};
		\node [style=Z dot] (29) at (4.5, 4) {};
		\node [style=X dot] (30) at (8.5, 3) {};
		\node [style=Z dot] (31) at (8.5, 4) {};
		\node [style=X dot] (32) at (7.75, 2) {};
		\node [style=X dot] (33) at (7, 0) {};
		\node [style=Z dot] (34) at (7.75, 4) {};
		\node [style=Z dot] (35) at (7, 4) {};
		\node [style=Z phase dot] (36) at (2, -1) {$\ -\frac\pi8~$};
		\node [style=none] (37) at (9.5, -1) {};
		\node [style=Z phase dot] (38) at (2, -2) {$\ -\frac\pi4~$};
		\node [style=none] (39) at (9.5, -2) {};
		\node [style=Z phase dot] (40) at (2, -3) {$\ -\frac\pi2~$};
		\node [style=none] (41) at (9.5, -3) {};
		\node [style=Z phase dot] (42) at (2, -4) {$\pi$};
		\node [style=none] (43) at (9.5, -4) {};
		\node [style=large box, minimum height=4.6cm] (44) at (5.75, 0) {Add};
		\node [style=Z phase dot] (45) at (-3.75, -1) {$\ -\frac\pi8~$};
		\node [style=none] (46) at (-1.25, -1) {};
		\node [style=Z phase dot] (47) at (-3.75, -2) {$\ -\frac\pi4~$};
		\node [style=none] (48) at (-1.25, -2) {};
		\node [style=Z phase dot] (49) at (-3.75, -3) {$\ -\frac\pi2~$};
		\node [style=none] (50) at (-1.25, -3) {};
		\node [style=Z phase dot] (51) at (-3.75, -4) {$\pi$};
		\node [style=none] (52) at (-1.25, -4) {};
		\node [style=none] (53) at (0, 2) {\eqref{eq:phase-gate-on-ancillae}};
		\node [style=none] (54) at (0, 1) {\eqref{eq:cat-addition-circ}};
	\end{pgfonlayer}
	\begin{pgfonlayer}{edgelayer}
		\draw (0.center) to (5);
		\draw (5) to (3);
		\draw (1) to (2);
		\draw (8.center) to (10.center);
		\draw (17) to (18);
		\draw (21) to (22);
		\draw (24) to (25);
		\draw (28) to (19);
		\draw (26) to (29);
		\draw (31) to (30);
		\draw (34) to (32);
		\draw (33) to (35);
		\draw (6.center) to (5);
		\draw (37.center) to (36);
		\draw (39.center) to (38);
		\draw (41.center) to (40);
		\draw (43.center) to (42);
		\draw (46.center) to (45);
		\draw (48.center) to (47);
		\draw (50.center) to (49);
		\draw (52.center) to (51);
	\end{pgfonlayer}
\end{tikzpicture}
\end{equation}
This is the construction that is presented in~\cite{Gidney2018halvingcostof}. There it was proven correct by arguing about the interaction between the quantum Fourier transform and addition. In comparison, the construction we present here is more bottom-up and only uses elementary facts about quantum circuits and catalysis.

\end{document}